\pdfoutput=1
\documentclass{acm_proc_article-sp}

\usepackage[
  pass,% keep layout unchanged 
]{geometry}

\usepackage{hyperref}
\usepackage{mathtools}
\usepackage{booktabs}
\usepackage{graphicx,color}

\usepackage{algorithm}
\usepackage[noend]{algorithmic}

\usepackage{ntheorem}

\usepackage{caption}
\usepackage{subcaption}

\makeatletter
\g@addto@macro{\UrlBreaks}{\UrlOrds}
\makeatother

\theoremstyle{plain}
\newtheorem{definition}{Definition}

\theoremstyle{plain}
\newtheorem{theorem}{Theorem}

\theoremstyle{plain}

\theoremstyle{plain}

\theoremstyle{plain}
\newtheorem{lemma}{Lemma}

\theoremstyle{plain}

\newcommand{\mc}[1]{\mathcal{#1}}

\begin{document}\sloppy

\title{On the Unicity of Smartphone Applications}

\numberofauthors{3} 
\author{
\alignauthor
Jagdish Prasad Achara\\
       \affaddr{INRIA}\\
       \email{jagdish.achara@inria.fr}
\alignauthor
Gergely Acs\\
       \affaddr{INRIA}\\
       \email{gergely.acs@inria.fr}
\alignauthor 
Claude Castelluccia\\
       \affaddr{INRIA}\\
       \email{claude.castelluccia@inria.fr}
}

\maketitle

\begin{abstract}
Prior works have shown that the list of apps installed by a user reveal a lot about user interests and behavior.
These works rely on the semantics of the installed apps and show that various user traits could be  learnt automatically using off-the-shelf machine-learning techniques. 
In this work, we focus on the re-identifiability issue and thoroughly study the unicity of smartphone apps on a dataset containing 54,893 Android users collected over a period of 7 months.
Our study finds that any 4 apps installed by a user are enough (more than 95\% times) for the re-identification of the user in our dataset.
As the complete list of installed apps is unique for 99\% of the users in our dataset, it can be easily used to track/profile the users by a service such as Twitter that has access to the whole list of installed apps of users.
As our analyzed dataset is small as compared to the total population of Android users, we also study how unicity would vary with larger datasets.
This work emphasizes the need of better privacy guards against collection, use and release of the list of installed apps.

%People are unique in the way they are and in the way they behave.  
%%In fact, people personalities are largely influenced by the environment they interact with, their education or interests.  
%For example, two recent studies showed that 4 or 5 location points or buying items are enough to uniquely identify a user in a large population.
%This paper extends this line of research by studying the unicity of mobile apps installed by users on their smartphone.

%The list of apps of a users can reveal a lot about user interests and behavior. 
%It could actually be used to profile users. 
%As a matter of fact, Twitter recently announced that it will use this information for targeted advertising. 
%But how unique are the lists the apps installed by a user? 
%Could they be used as  identifiers? 
%This paper studies the unicity of installed applications using a dataset of 56,000 Android users collected over a period of 7 months. 
%We show that, in more than 95\%, any 4 apps installed by a user is unique and can be used as an identifier.
%Therefore, a service, such as Twitter, that has access to the list of apps of a users, can easily use this information
%to track users. We also analyse how this unicity varies with the size of the datasets, and show trat this result also hold for much larger datasets. 
%This work emphasizes the need of better privacy safeguards against collection, use and release of the list of installed apps.
\end{abstract}
%\category{K.4}{Public Policy Issues}{Privacy}

%\terms{Re-identifiability; Smartphone applications; Unicity; Privacy; User profiling; Uniqueness}
\vspace{1mm}
\noindent
{\bf Categories and Subject Descriptors:} K.4 {[Public Policy Issues]}: {Privacy}

%\vspace{1mm}
%\noindent
%{\bf Keywords:} Re-identifiability; Smartphone applications; Unicity

%%%%%%%%%%%%%%%%%%%%%%%%%%%%%%%%%%%%%%%%%%%%%%%%%%%%%%%%
%=======================================================
%%%%%%%%%%%%%%%%%%%%%%%%%%%%%%%%%%%%%%%%%%%%%%%%%%%%%%%%
\section{Introduction}
\label{sec:intro}

%The personality that defines a person is directly influenced by the environment he/she interacts with. 
%As a child you learn by seeing your peers and your elders. 
%Then there is the question of the different languages and cultures that you are brought up in. 
%And then there is the media. 
%We form our own prejudices and these shape our personality and makes it unique. 
%We might be able to imitate another person but can never replace them.

People are all unique the way they are or they look. 
They can be easily identified from their DNA sequences, fingerprints, Iris scans, web browsers and so on. 
Also, a combination of various attributes, such as their age, address or religion \cite{Sweeney2001} makes them unique.
Recently, some studies have shown that people are also unique in the way they behave.
For example, de Montjoye \emph{et al.} illustrated this \emph{behavioural uniqueness} by showing that people are unique in the way they move around~\cite{Nature13}. 
In fact, they show that only four spatio-temporal positions are enough to uniquely identify a user 95\% of the times in a dataset of one and a half million users. 
%It has also been shown recently that people are unique in the way they purchase goods online \cite{Science15}. 
%A recent study demonstrated that the dates and locations of four purchases are enough to identify 90 percent of the people in a data set recording three months of credit-card transactions by 1.1 million users \cite{}.  
Similarly, other studies showed that people are unique in the way they purchase goods online~\cite{Science15} or configure their browser~\cite{Eckersley10} or browse the web~\cite{Olejnik13}. 
%Out of 223,000 web histories, 98 percent of the web history profiles were unique \cite{}.

%Uniqueness of an individual is also influenced by the day to day happenings around him or her. 
%Some recent studies have shown that people are also unique in the way they behave.
%For example, de Montjoye \emph{et al.} \cite{} illustrated
%this \emph{behavioural uniqueness} by showing that people are unique in the way they move. In fact,
%they show that only four or five location positions are enough to identify a user in a set of 10 millions users \cite{}. 
%It was also recently shown that people are unique in the way they consume and buy things. A recent study demonstrated that
% the dates and locations of four purchases are enough to identify 90 percent of the
%people in a data set recording three months of credit-card transactions by 1.1 million users \cite{}.  Similarly, another study showed that
%people browsing habits or patterns are unique. Out of 223,000 web histories, 98 percent of the profiles were unique \cite{}.

%As smartphones have been widely deployed/adopted today all over the world and the list of installed/running applications (apps) on them is readily available to be accesssed by apps (or included third-party libraries), the threat in terms of user privay is huge.

As smartphones have been widely deployed today all over the world and the list of installed/running applications (apps) on them is readily available to be accessed, the threat in terms of user privacy is huge if this data is collected, used and released without sufficient diligence in terms of privacy.
This threat comes in two flavors: first, the semantics of the installed apps can tell a lot about the users' habits and interests \cite{Seneviratne:2014}, and second, the unicity of installed apps could make a user re-identifiable if this dataset is released.
In fact, regarding the first privacy threat, \cite{Seneviratne:2014} showed that user traits such as religion, relationship status, spoken languages, countries of interest, and whether or not the user is a parent of small children, can be easily predicted from the list or even the categories of the installed apps on smartphones.
In this paper, we focus on the second privacy threat and measure the unicity of installed apps to be able to measure the risk of re-identification if app datasets are released in public or shared between two entities.

It is quite in the news these days
%\footnote{\url{http://recode.net/2015/06/10/twitter-advertisers-can-now-target-you-based-on-the-other-apps-on-your-phone/}}
\footnote{\url{http://goo.gl/00FbD0}}
that Twitter has started to collect the list of apps that a user has installed.
They claim to use this information for targeted interest-based advertising among others.
However, it might be a privacy concern if Twitter shares this list of installed apps with an advertising company, even in pseudo-anonymized form, i.e., after removing all direct user identifiers (and even if app names are replaced with their hashes).
This is because the advertising company might implicitly know a subset, say $K$, of installed apps of a user in which their ad library is present.
So if $K$ apps are enough to uniquely identify a user in the dataset, the advertiser would be able to re-identify the user in the Twitter dataset, and hence, learn about all the other installed apps of that user.
By knowing this whole list of installed apps of a user, the advertiser can learn about that user's interests and habits (as demonstrated in \cite{Seneviratne:2014}), and consequently, might be able to deliver the targeted ads directly in these apps in which its library is present. %and are installed by that particular user.
We believe that this is a real privacy threat to smartphone users (both Android and iOS) today as apps running on these OSs can access the list of installed/running apps.
It is to be noted that Android apps do not require any permission to access the whole list of installed apps.
On iOS, Apple does not provide a public API to access the list of installed apps but apps can get the list of currently running apps at any time.
And if an app makes a frequent scan of currently running apps over a period of time, the list can converge very fast to the list of installed apps.

\textbf{Contributions:} The main contributions are as follows.
\begin{itemize}
\item We show that 99\% of the lists of installed applications of users are unique out of a total of 55 thousands users. 
Moreover, as few as two applications are sufficient for an adversary to identify an individual's application list with a probability of 0.75 in our dataset. 
The re-identification probability increases to almost 0.95 if the adversary knows 4 apps.
We stress that these results were obtained without considering any system apps (which are common for all users), and apps were identified only by the hash of their names.
Incorporating additional information into their identifiers, such as app version, time of installation, etc., would increase these probabilities even more. 

\item We propose an unbiased estimate of the real uniqueness of any subset of applications, i.e., the probability that a randomly selected subset of apps with cardinality $K$ is unique in the dataset. 
For this purpose, we use a Markov Chain Monte Carlo method to sample subsets of applications from a dataset uniformly at random. 
We prove that this chain is generally fast-mixing with most practical datasets, i.e., has a running time complexity which is roughly linear in the dataset size and $K$. 
This result might be of independent interest, as this technique can be used to sample subsets with arbitrary cardinality from any set-valued dataset. 

\item We attempt to predict the uniqueness of lists of applications in larger datasets using standard non-linear regression analysis. 
Our learned model performs well on our limited app dataset as well as on mobility data with sufficiently large number of users.
However, in case of mobility data, we find that the model is not able to predict well if it is trained with smaller datasets, e.g., of the size of our app dataset.
Therefore, we conclude that our app dataset at hand might probably be too small to accurately predict the uniqueness in a larger datasets such as all Android users worldwide. 
\end{itemize}

\section{Unicity as a measure of re-identifiability}

Let $\mathbb{A}$ denote the universe of all apps, where each application is represented by a unique identifier in $\mathbb{A}$. A dataset $D\subseteq 2^\mathbb{A} \setminus \{\}$ is the ensemble of all apps of some set of individuals, where $|D|$ denotes the number of individuals in $D$. A record $D_u$, which is a non-empty subset of $\mathbb{A}$, refers to all apps of an individual $u$ in $D$. 
A set of applications with cardinality $K$ is shortly called $K$-apps henceforth. The set of all $K$-apps over $\mathbb{A}$ is denoted as $\mathbb{A}^K$.

\begin{definition}[Unicity]
\label{def:unicity}
Let $\mathit{supp}(x,D)$ denote the support of $x \in \mathbb{A}^K$ in dataset $D$, i.e., the number of records in $D$ which contain $x$. Then, 
$$
H_1 = \frac{|\{x : x \in \mathbb{A}^K \wedge \mathit{supp}(x,D)=1\}|}{|\{x : x \in \mathbb{A}^K \wedge \mathit{supp}(x,D) \geq 1\}|}
$$
is defined as the unicity (or uniqueness) of $K$-apps in $D$.
\end{definition}

%the relationship between different sizes of sets of apps of a user and the number of users they are found in.
%For example, if we consider sets of size 2, we find what fraction of total such sets of size 2 are found in any single user or any 2 users etc.
%One specific case of this relative abundance distribution is when the number of users considered is 1, i.e., a single user.
%This specific case of relative abundance distribution is termed as \emph{Unicity} and is studied because of its privacy importance.
The unicity of $K$-apps is the relative frequency of $K$-apps which are contained by only a single record. In general, \emph{relative abundance distribution} (RAD)\footnote{This term is often used in the field of ecology 
to describe the relationship between the number of observed species as a function of their observed abundance.} is a relative frequency histogram $\mathbf{H} = (H_1, H_2, \ldots, H_n)$ of $K$-apps with respect to a dataset $D$, where $H_i$ denotes the relative frequency of $K$-apps which are contained by exactly $i$ records in $D$, i.e., 
$
H_i = \frac{|\{x : x \in \mathbb{A}^K \wedge \mathit{supp}(x,D)=i\}|}{|\{x : x \in \mathbb{A}^K \wedge \mathit{supp}(x,D) \geq 1\}|}
$.

Unicity is strongly related to re-identifiability, and we use it as a measure of privacy in this paper:
it is the probability that an adversary, who only knows $K$ applications installed on a user's device, can single out the record of this user in $D$. 
Indeed,  any $K$-apps which is unique in $D$ can be used as a personal identifiable information (PII) of its record owner. Specifically, if the adversary knows such $K$-apps, it can easily identify the corresponding record and retrieve all the applications installed by its owner, even if $D$ is pseudo-anonymized (i.e., does not contain any direct PII such as device ID or personal name).  
Therefore, large unicity usually indicates a serious privacy risk in practice. 

%The more a particular data or observation is unique, more privacy risks it carries, i.e., unicity is directly related to the re-identifiability.

%\emph{Unicity} is a measure to indicate that exactly one object with a certain property exists.
%In our case, a set of apps of a certain size is unique if there is exactly one user in which they are found together.

%%%%%%%%%%%%%%%%%%%%%%%%%%%%%%%%%%%%%%%%%%%%%%%%%%%%%%%%
%=======================================================
%%%%%%%%%%%%%%%%%%%%%%%%%%%%%%%%%%%%%%%%%%%%%%%%%%%%%%%%
\section{Approximating unicity with sampling}
\label{sec:sampling}

To compute unicity, and RAD in general, the support of all different $K$-apps in $D$ should be calculated. 
However, this is usually prohibitively expensive in practice.
Therefore, like previous works \cite{Science15, Nature13}, we rely on sampling to estimate unicity. 
In particular, let $\Omega^K$ denote the set of all $K$-apps which occur in at least one individual's record, i.e., $\Omega^K = \{x : x \in \mathbb{A}^K \wedge \mathit{supp}(x,D) \geq  1\}$.
We randomly sample a set $V$ of $K$-apps from $\Omega^K$, and approximate the real unicity $H_1$ by the sample unicity $\hat{H}_1 = 
 \frac{|\{x : x \in V \wedge \mathit{supp}(x,D)=1|}{|V|}$, where $V \subseteq \Omega^K$ is the sample set, and $n=|V|$ is the sample size.

\subsection{Biased vs. unbiased estimation of unicity}
\label{sec:biased}
How should we sample $K$-apps from the dataset? A popular technique, which has been used in several works \cite{Science15, Nature13},  first samples a user uniformly at random in $D$, and then a set of $K$ applications from this user's record also uniformly at random. However, this simple technique provides a \emph{biased estimation} of the unicity in Definition \ref{def:unicity}, if the estimator remains the sample mean $\hat{H}_1$, since $E[\hat{H}_1] \neq H_1$.  
In fact, $K$-apps which occur in more records of $D$ become more likely to be selected by this approach (assuming records have similar sizes). 
%As a result, the samples are not selected uniformly at random over all possible $K$-apps in the dataset. 
As a result, this sampling method is biased towards more popular $K$-apps, and the measured unicity is an underestimation of the real unicity $H_1$ what one would get with an unbiased estimator of $H_1$. Such an unbiased estimator can be the sample unicity $\hat{H}_1$ of $K$-apps which are sampled truly uniformly at random from $D$. 
This is also illustrated in Figure~\ref{fig:unicity_with_sampling_type}, where the sample unicity of biased and unbiased (i.e., uniform) samples are reported.

Before describing our unbiased estimation of unicity $H_1$, we shed some light on the privacy semantics behind the two sampling approaches.  
The biased technique approximates the success probability of an adversary who is more likely to know popular $K$-apps from the application set of any user. 
For instance, continuing the the case of the advertiser from Section \ref{sec:intro}, the advertiser's library should be more likely to be used by popular apps (such as Facebook, Twitter, etc.), 
which are installed on many devices, rather than by other less popular apps. 
However, this is not necessarily true and in general, an advertiser's library can be included in any $K$-apps of a user.
% (e.g., the adversarial library is favoured by certain applications) with larger probability, which changes the sampled unicity $\hat{H}_1$ accordingly; 
% more popular apps tend to increase $\hat{H}_1$ while unpopular apps tend to decrease it yielding a biased estimation of $H_1$. 
In the rest of the paper, we assume that the adversary can learn \emph{any} $K$-apps of \emph{any} users in $D$ with equal probability, which is the most general assumption in practice. 
Therefore, we  are interested in an unbiased estimator of $H_1$. 
%Conversely, if the adversary was more likely to learn unpopular $K$-apps of users, the unique $K$-apps should have larger weight in the sampling process, and the sample mean will be an overestimation of the unicity in Definition \ref{def:unicity}. We do not see any real justification for any of these two extreme adversarial models in practice, and hence we assume that the adversary can learn any $K$-apps in the dataset with equal probability. For example, there is no guarantee that a particular library collecting the set of installed apps are favoured by popular or unpopular applications on the AppStore. %In fact, the unicity computed over uniformly sampled $K$-apps approximates the success probability of a much stronger adversary who can know $\emph{any}$ $K$-apps of \emph{any} users with equal probability.  
%Therefore, the unicity of non-biased samples provide a \emph{worst-case} measure of re-identifiability in $D$.

\subsection{Uniform sampling of $K$-apps}
\label{sec:uniform_sampling}

 A unbiased estimation of $H_1$ is obtained, if $\hat{H}_1$ is computed over a sample set where each $K$-apps can appear with equal probability.
Hence, our task is to sample an element from $\Omega^K$ uniformly at random for any $K$. A first (naive) approach could be to use rejection sampling, i.e., sample a candidate $K$-apps from $\mathbb{A}^K$ uniformly at random, and then accept this candidate as a valid sample only if it also occurs in $D$. Otherwise, repeat the process until a candidate is accepted.  Although sampling a candidate from $\mathbb{A}^K$ is straightforward, it is very likely to be non-existent in $D$ (especially if $K$ is large), and hence, its running complexity is $O(|\mathbb{A}|^K)$ in the worst case. An alternative approach could be to enumerate $\Omega^K$, and choosing one element directly from $\Omega^K$ uniformly at random. However, the complexity of this approach is still $O(|D|(\max_u|D_u|)^K/K!)$. Unfortunately, these naive methods provide acceptable performance only if $K$ is small. As Table \ref{tab:D} shows, in our dataset, $|\mathbb{A}| = 92210$, $\max_u|D_u| = 541$, $|D| = 54893$, and we wish to estimate the unicity when $1\leq K \leq 10$. 

We instead propose a  sampling technique based on the Metropolis-Hastings algorithm \cite{MRRTT53jcp, Chib95}, which is a Markov Chain Monte Carlo (MCMC) method. Our proposal has a worst-case complexity of only $O(K|D|/H_1^*)$, where $H_1^*$ is roughly the unicity of $K$-apps in $D$.  As the unicity of $K$-apps is large, especially if $K$ is large, the complexity is approximately $O(K|D|)$ in practice. Hence, our sampling technique remains reasonably fast even for larger values of $K$.

In particular, we construct an ergodic Markov chain, denoted by $\mathcal{M}$,
such that its stationary distribution $\pi$ is exactly the distribution that we want to sample from, that is, the uniform distribution over $\Omega^K$.  Each $K$-apps in $\Omega^K$ corresponds to a state of $\mc{M}$, and we simulate $\mc{M}$ until it gets close to $\pi$, at which point the current state of $\mc{M}$ can be considered as a sample from $\pi$. 
$\mathcal{M}$ is detailed in Algorithm \ref{alg:generic_chain}. At each state transition, $\mc{M}$ picks a candidate next state $C$ independently of the current state $S$ (in Line 6-7). In Line 8, the candidate is either accepted  (and $\mc{M}$ moves to $C$) or rejected with certain probability (in which case the candidate state is discarded, and $\mc{M}$ stays at  $S$). 
The main idea is that, at each state, we use the fast but biased sampling technique, which is described in Section \ref{sec:biased}, to propose a candidate $C$ (in Line 6-7). We correct this bias by adjusting the acceptance/rejection probability (in Line 8) accordingly; $\mathcal{M}$ is more likely to accept such $K$-apps which are less likely to be proposed in Line 6-7. Indeed, as  $\pi(S) = \pi(C)$, the probability of acceptance is
$\min\left(1,  \frac{ Pr[\text{``$S$ is proposed''}]}{Pr[\text{``$C$ is proposed''}]}\right) = \min\left(1, \frac{\sum_{\forall u : U_u \supseteq S}  1 / {|U_u| \choose K}}{ \sum_{\forall u : U_u \supseteq C}  1 / {|U_u| \choose K}}\right) = \min \left(1, q(S)/q(C)\right)$. A more formal analysis is described in Appendix \ref{sec:app}. The proofs of all the theorems in this paper can be found in Appendix \ref{sec:app}.

\begin{algorithm}[h!]
\begin{algorithmic}[1]
\small
\STATE {\bf Input:} Dataset $D$, $K$, \# of iterations $t$ \\
\STATE {\bf Output:} A sample $S \in \Omega^K$\\
\STATE Let $U : = \{D_u : |D_u| \geq K \wedge D_u \in D\}$
\STATE Let $S$ be an arbitrary $K$-apps in $\Omega^K$ 
\FOR  {$k=1$ \textbf{to} $t$}
\STATE Select an individual $u \in [1,|U|]$ uniformly at random 
\STATE Select a subset $C \subseteq U_u$ uniformly at random such that $|C| = K$
\STATE Let $S := C$ with probability $\min\left(1, q(S)/q(C)\right)$, where $q(x) = \sum_{\forall u : D_u \supseteq x} \prod_{i = 1}^{K} \frac{1}{|U_u| - K + i} $
\ENDFOR
\STATE \textbf{return} $S$ 
\end{algorithmic}
\caption{MCMC sampling ($\mc{M}$)} \label{alg:generic_chain}
\end{algorithm}  

\begin{theorem}
\label{THM:MARKOV}
$\mathcal{M}$ in Algorithm \ref{alg:generic_chain} is an ergodic Markov chain whose unique stationary distribution is the  uniform distribution over $\Omega^K$ for any $K$.
\end{theorem}

\begin{figure*}
        \centering
        %\begin{subfigure}[b]{0.3\textwidth}
         %       \includegraphics[width=\textwidth]{{convergence_K_4}.pdf}
         %       \caption{$K=4$}
        %\end{subfigure}
	%\hspace{-0.70cm}
        ~ %add desired spacing between images, e. g. ~, \quad, \qquad, \hfill etc.
        \begin{subfigure}[b]{0.32\textwidth}
                \includegraphics[width=\textwidth]{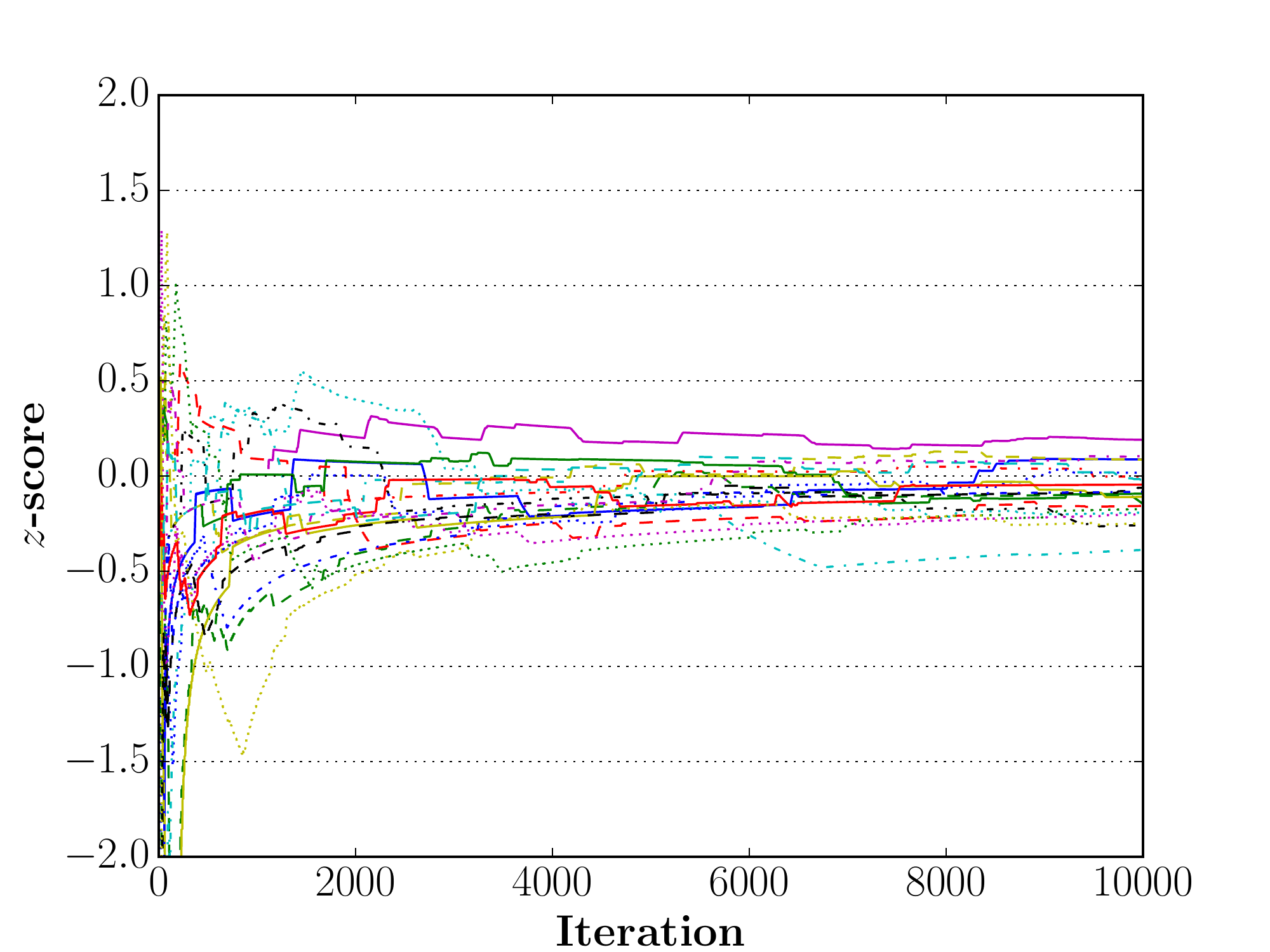}
                \caption{$K=5$}
        \end{subfigure}
        ~ %add desired spacing between images, e. g. ~, \quad, \qquad, \hfill etc.
        \begin{subfigure}[b]{0.32\textwidth}
                \includegraphics[width=\textwidth]{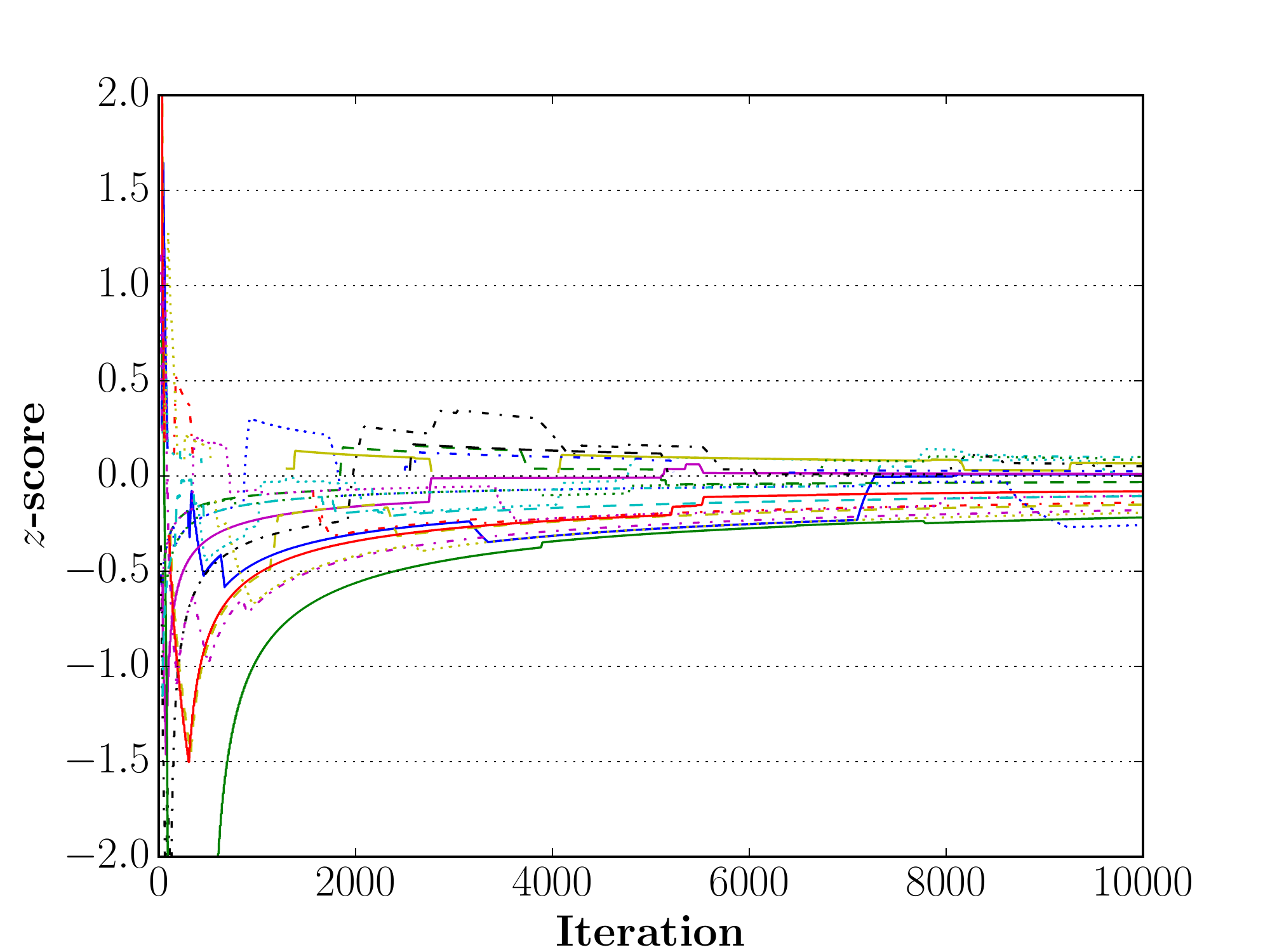}
                \caption{$K=6$}
        \end{subfigure}
	%\hspace{-0.45cm}
        \begin{subfigure}[b]{0.32\textwidth}
                \includegraphics[width=\textwidth]{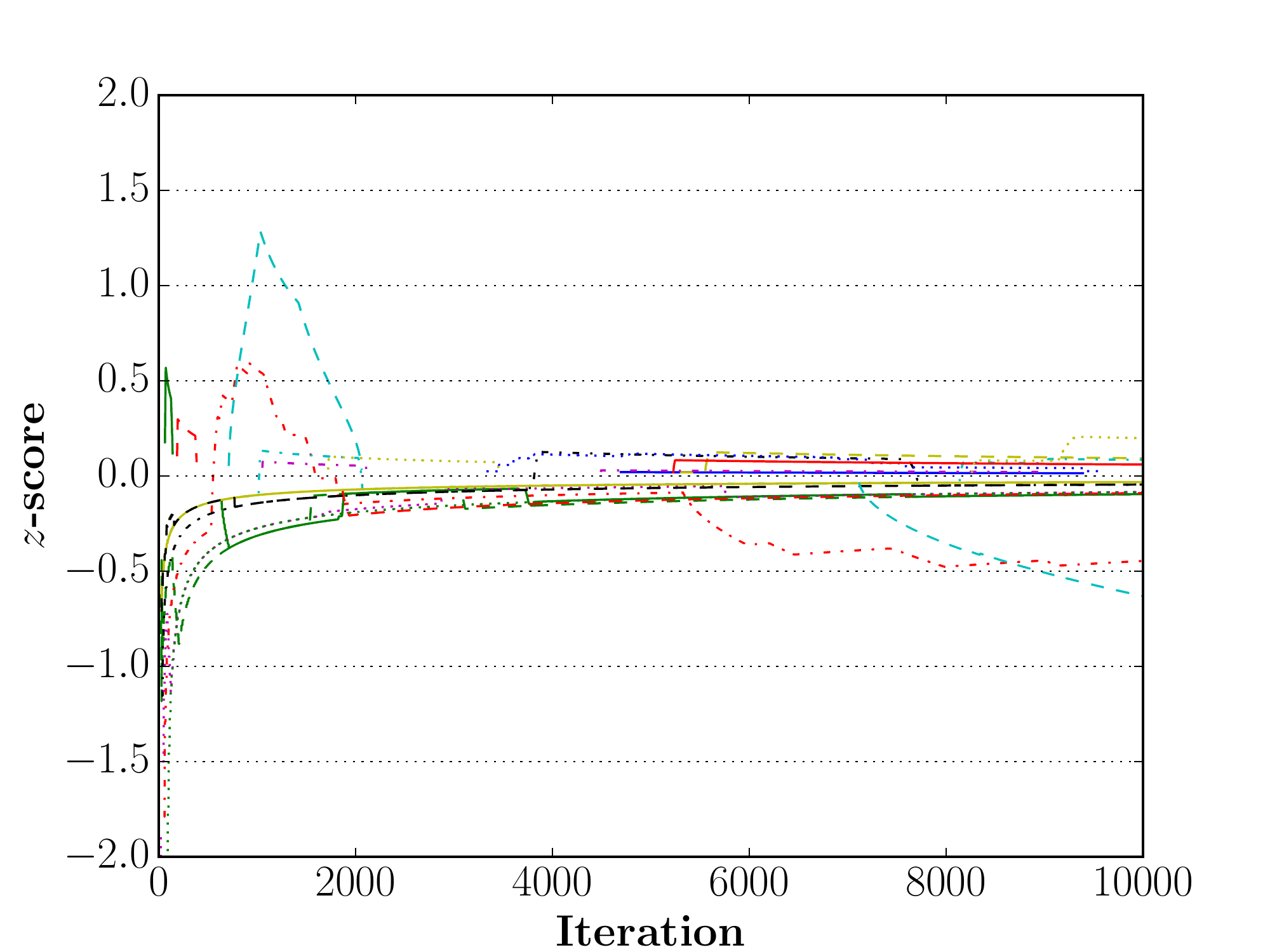}
                \caption{$K=7$}
        \end{subfigure}       
 \caption{\label{fig:conv}Convergence of our Markov chain $\mathcal{M}$. The $z$-score, depending on the number of iterations $t$,  of 20 independent chains   are plotted.}
\end{figure*}

\paragraph{Convergence of $\mathcal{M}$}
How to adjust $t$ in Algorithm \ref{alg:generic_chain}?
%How many transitions $\mc{M}$ needs to do in order to be ``close enough" to $\pi$? 
We prove that a ``good" uniform sample from $\Omega^K$ can be obtained roughly after $O(|D|)$ iterations in most practical cases.
%The fundamental question is how to adjust $t$ in Algorithm \ref{alg:generic_chain}, i.e., how many transitions $\mc{M}$ needs to do in order to be ``close enough" to $\pi$. 
The time that $\mc{M}$ takes to converge to its stationary distribution $\pi$ is known as the \emph{mixing time} of $\mc{M}$, and is measured in terms of the total variation distance between the distribution at time $t$ and $\pi$. 

\begin{definition}[Mixing time]

For $\xi > 0$, the mixing time $\tau_{\mc{M}}(\xi)$ of Markov chain $\mc{M}$ is
$$
\tau_{\mc{M}}(\xi) = \min\{ t' : ||P^t_{\mc{M}} - \pi||_{\mathit{tv}} \leq \xi, \forall t \geq t'\}
$$
where 
$
||P^t_{\mc{M}} - \pi||_{\mathit{tv}} =\max_{x\in \Omega^K}\frac{1}{2}\sum_{y\in \Omega^K}| P^t_{\mc{M}}(x,y) - \pi(y)|
$ defines the total variation distance. $P^t_{\mc{M}}(x,y)$ denote the $t$-step probability of going from state $x$ to $y$, and
$P^t_{\mc{M}}$ denote the $t$-step probability distribution over all states.
\end{definition}

The next theorem shows that $\mathcal{M}$'s mixing time is $O(|D|\log(1/\xi)/H_1^*)$, where $|D|$ is the dataset size and $H_1^*$ is the unicity of $K$-apps from the largest record of $D$. 
As the unicity of $K$-apps is usually large in practice, especially if $K$ is large, $\mathcal{M}$ is fast-mixing in general. In our dataset $D$, $0.6 \leq H_1^* \leq 0.999$ for $2\leq K \leq 9$ \footnote{The unicity of $K$-apps from a single record can easily be approximated with Inequality \ref{eq:sampling_bound} using uniform samples over all $K$-apps from the record. Likewise the biased sampling in Section \ref{sec:biased}, this sampling is easy to implement (e.g., by choosing $K$ items individually from the record without replacement).}.

\begin{theorem}[Mixing time of $\mathcal{M}$] 
\label{THM:MIXING}
Let $H_1^*$ denote the probability that a randomly selected set of $K$ items from the largest record  (i.e., having the most apps) in $D$ is unique. Then,
$\tau_{\mathcal{M}}(\xi) \leq |D|\ln(1/\xi)/H_1^*$ for any $K$. 
\end{theorem}

We emphasize that the bound in Theorem \ref{THM:MIXING} is a worst-case bound, and the real convergence time can be much smaller depending on the dataset $D$ as well as the starting state of the chain. As we show next, $\mathcal{M}$ indeed exhibits much smaller convergence time than its theoretical worst-case bound for our dataset. 
We detected the convergence of $\mathcal{M}$ using the Geweke diagnostic \cite{Geweke92}; if $X_t$ denotes a Bernoulli random variable describing whether the current state of $\mathcal{M}$ at time $t$ is unique, and $\mathbf{X}_t = (X_1, X_2, \ldots, X_t)$, then we compute the $z$-score  
$z= \frac{E[\mathbf{X}_a] - E[\mathbf{X}_b]}{\sqrt{\mathit{Var}(\mathbf{X}_a) + \mathit{Var}(\mathbf{X}_b)}}$, where $\mathbf{X}_a$ is the prefix of $\mathbf{X}_t$ (first 10\%), and $\mathbf{X}_b$ is the suffix of $\mathbf{X}_t$ (last 50\%). 
 We declare convergence when the $z$-score falls within $[-1,1]$. Indeed, if $\mathbf{X}_a$ and $\mathbf{X}_b$ become identically distributed (i.e., $\mathbf{X}_a$ and $\mathbf{X}_b$ appear to be uncorrelated), the $z$ values become normally distributed with mean 0 and variance 1 according to the law of large numbers. 
We simulated 20 instances of $\mathcal{M}$ each starting at different states, and plotted the $z$-score of each chain depending on the number of iterations $t$ in Figure \ref{fig:conv}.
This shows that convergence is detected roughly after  3000 steps in all chains with different values of $K$.
When this happens, the current state can be taken as a valid sample. Hence, in the sequel, we run $\mathcal{M}$ with $t=3000$ to obtain a uniform sample from $\Omega^K$.

We note that $q$ in Algorithm \ref{alg:generic_chain} can be computed rapidly in practice by precomputing another dataset $T$, where each record corresponds to an application in $D$, and record $i$ contains the sorted list of all users who have application $i$ in their record. Hence, the set of users who have a common specific $K$-apps can be  computed easily  by taking the intersection of the corresponding records in $T$. The complexity of this operation is $O(K|i_{\max}|)$, where $|i_{\max}|$ is the maximum record size in $T$, i.e., the number of users of the most popular application in $D$. Fast implementations of the intersection of sorted integers are described in \cite{LemireBK14}. 

%Supposing that $\ln(1/\xi) \ll i_{\max} \ll |D|$ and $1/H_1^*$ is a constant close to 1, the total running complexity of $\mathcal{M}$ is roughly $O(K |D|)$.

\subsection{Computing the sample size}

In order to compute the sample size, we use the Chernoff-Hoeffding inequality \cite{Hoeffding1963} on the tail distribution of the sum of independent (but not necessarily identically distributed) Bernoulli random variables. In particular, if $X_i$ denotes a Bernoulli random variable describing the event that the $i$th sampled $K$-apps is unique in $D$, then the deviation of the estimator $\hat{H}_1 = \sum_{i=1}^{n} X_i / n $ from $E[\hat{H}_1] = H_1$ is given by 
$
Pr\left[\left|\hat{H}_1  - H_1\right|\geq \varepsilon\right] \; 	 \leq \; 	2e^{-2n\varepsilon^2}
$
, or equivalently,
\begin{equation}
\label{eq:chernoff}
Pr\left[\left|\hat{H}_1  - H_1\right|< \varepsilon\right] \; 	 \geq \; 	1 - 2e^{-2n\varepsilon^2}
\end{equation}
where $\varepsilon$ is the sampling error and $\sigma = 1 - 2e^{-2n\varepsilon^2}$ is the confidence. Hence, we obtain that 
\begin{align}
\label{eq:sampling_bound}
n \geq \frac{1}{2\varepsilon^2} \ln\left(\frac{2}{1-\sigma}\right)
\end{align}
For example, if $\varepsilon = 0.01$ and $\sigma = 0.99$, we need to sample at least $26492$ $K$-apps from $D$ (with replacement) to guarantee that $|\hat{H}_1 -\hat{H}_1| < 0.01$ with probability at least $0.99$.
%This means that to calculate the unicity of any combination of 5 randomly chosen apps, we must calculate the fraction of any combination of 5 apps that are unique with respect to the total number of possible combinations of 5 apps.

Considering RAD, suppose we aim at approximating the first $k$ relative frequency values of $\mathbf{H}$, i.e., $(H_1, H_2, \ldots, H_k)$. Therefore, we wish to simultaneously satisfy Inequality \ref{eq:chernoff} for each $H_i$ $(1\leq i \leq k)$, where  $\hat{H}_i = \sum_{j=1}^{n} X'_j/ n $, and $X'_j=1$ if the $j$th sampled $K$-apps occurs in exactly $i$ records of $D$, otherwise $X'_j=0$. Hence,
\begin{align}
Pr\left[ \bigwedge_{i=1}^k \left|\hat{H}_i  - H_i\right|< \varepsilon\right] 	 &\geq 1 -  \sum_{i=1}^k Pr\left[\left|\hat{H}_i  - H_i\right|\geq \varepsilon\right] \notag \\	
&\geq  1 - 2ke^{-2n\varepsilon^2} \notag%{from Inequality \ref{eq:chernoff}}
\end{align}
where $\delta =  1 - 2ke^{-2n\varepsilon^2}$ is the confidence. Therefore, 
\begin{align}
\label{eq:rad_sample_size}
n \geq \frac{1}{2\varepsilon^2} \ln\left(\frac{2k}{1-\sigma}\right)
\end{align}
For instance, for $\varepsilon = 0.01$, $\sigma = 0.99$, and $k=10$, we need to sample at least $38005$ $K$-apps from $D$ (with replacement). This will guarantee that $|\hat{H}_i -\hat{H}_i| < 0.01$ for all $1 \leq i \leq k$ with probability at least $0.99$.

%Table~\ref{tab:confidence_error_samples} gives the number of samples required for different values of confidences $\sigma$ and errors $\epsilon$.
%\begin{table}[]
%  \centering
%\begin{tabular}{cc|c}
%\toprule
%\textbf{Confidence ($\sigma$)}    & \textbf{Error ($\epsilon$)} & \textbf{Required \# Samples} \\
%\midrule
%0.99    & 0.01  & 26,492        \\
%0.99    & 0.05  & 1,060          \\
%0.99    & 0.1   & 265           \\
%0.95    & 0.01  & 18,444         \\
%0.95    & 0.05  & 738           \\
%0.90    & 0.01  & 14,979         \\
%0.90    & 0.05  & 599           \\
%\bottomrule
%\end{tabular}
%\caption{Number of samples required for a particular confidence $\sigma$ and maximum possible error $\epsilon$}
%\label{tab:confidence_error_samples}
%\end{table}
%

%\begin{table}[h]
%\centering
%\begin{tabular}{|c|c|} \hline
%$K$ & \emph{Unicity} \\ \hline \hline
%2 & 0.59 \\ \hline
%3 & 0.88 \\ \hline
%4 & 0.97 \\ \hline
%5 & 0.99 \\ \hline
%6 & 0.998\\ 
%\hline
%\end{tabular}
%\end{table}

%%%%%%%%%%%%%%%%%%%%%%%%%%%%%%%%%%%%%%%%%%%%%%%%%%%%%%%%
%=======================================================
%%%%%%%%%%%%%%%%%%%%%%%%%%%%%%%%%%%%%%%%%%%%%%%%%%%%%%%%
\section{Evaluation}

\subsection{Dataset characteristics}
\label{sec:dataset}
The analyzed dataset comes from the Carat research project \cite{Oliner:2013}.
The dataset includes data from $54,893$ Carat Android users between March 11, 2013 and October 15, 2013  \cite{Truong:2014}.
During this period, the Carat app\footnote{\url{http://carat.cs.helsinki.fi}} was collecting the list of running apps (and not the list of all installed apps) on users' devices when the battery level changes.
%However, the newer versions of Carat apps now collect the whole list of installed apps and keep track of app (un)installation event.
\emph{As collecting the list of running apps multiple times over more than 7 months is likely to sum up to the set of all installed apps of a user, we consider a record as the set of installed applications in this paper, even if a record might not be the complete set of installed apps all the time.}
%Also, as system apps are common to all users, we do not consider them for our study, and therefore removed them from all records.

% The goal of the project is to inform users about the energy consumption profile of their running apps and give personalized recommendations for improving the battery life.
% To do so, users need to install the Carat app on their smartphones which intermittently takes measurements about the device, especially when the battery level changes.
% One type of information collected (among others) by the Carat app is the list of running apps on the device.

% The Carat app is available on both Android and iOS. 
% However, our dataset contains only Android users of the Carat app.
%The dataset is shared with us by the Carat team in a pseudo-anonymised form where each row (record) corresponds to a user.
%In particular, a record in this dataset $D$ contains the set of apps of a single user where the real package name of the app is replaced by its SHA1 hash.
%Additionally, Carat team also provided us with the SHA1 hashes of all system apps.
%The data sharing agreement requires us not to use the data for deanonymization of users present in this data.
%This is because Carat team assures users to respect their privacy requirements while also making the dataset useful for public research.

\begin{figure}[h]
	\centering
	\includegraphics[width=0.7\linewidth]{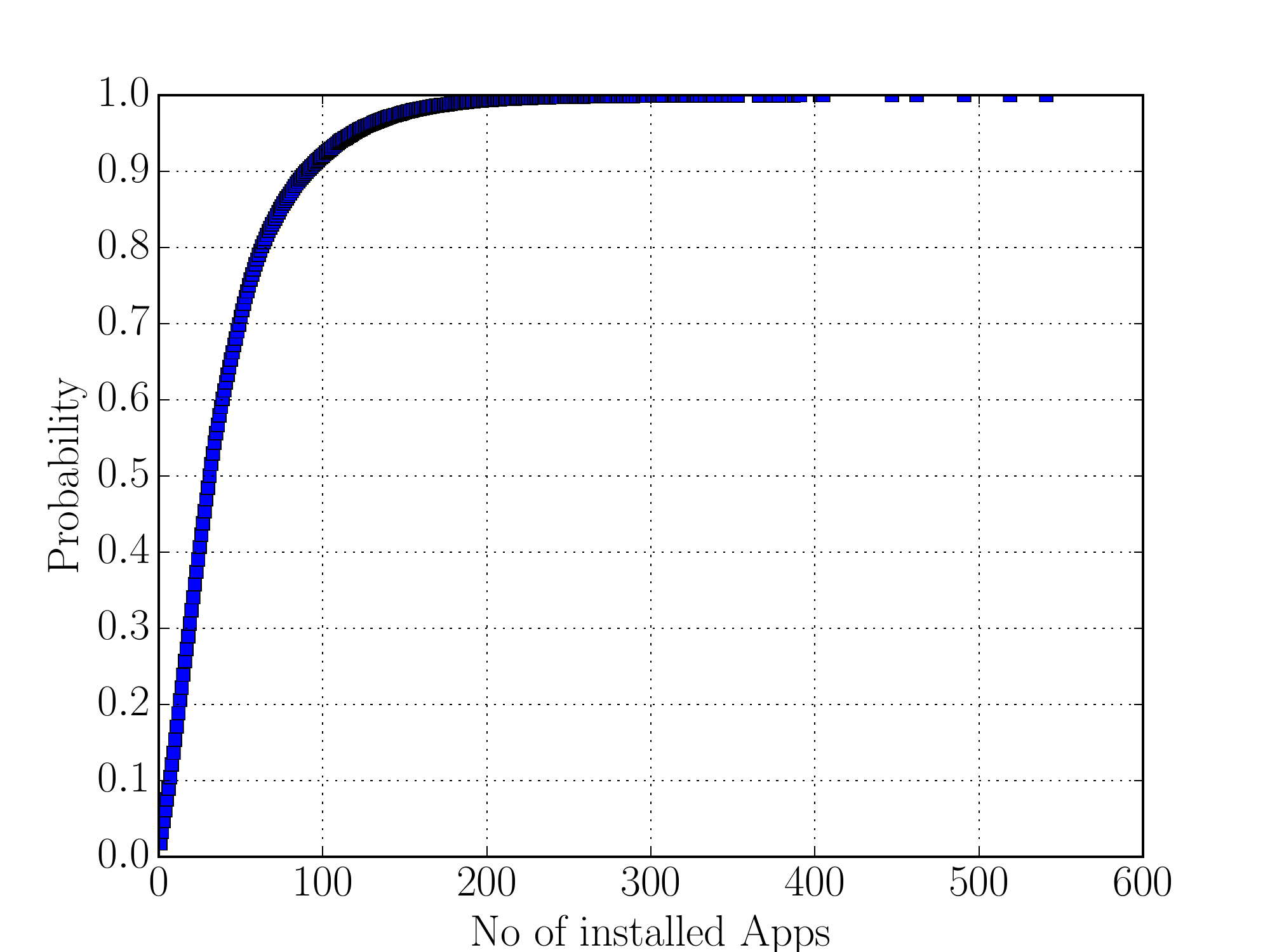}
	\caption{Cumulative distribution of installed apps}
	\label{fig:cdf}
\end{figure}

\begin{figure}[h]
	\centering
	\includegraphics[width=0.7\linewidth]{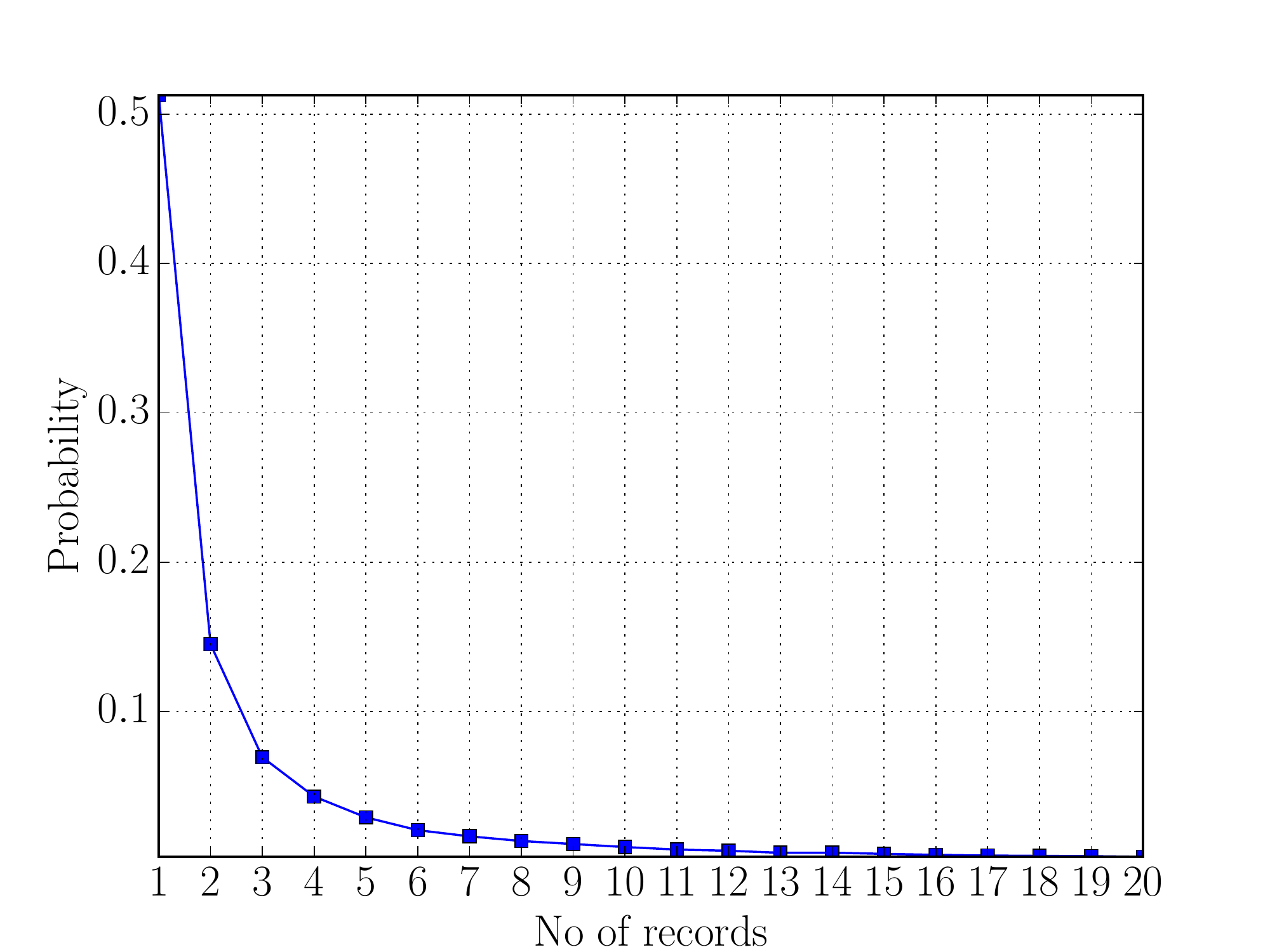}
	\caption{Probability distribution of number of records containing a particular app}
	\label{fig:pdf_users_cut}
\end{figure}

We removed system apps from all records because they are common to all users.
Without system apps, our analyzed dataset contains 92,210 different applications whereas total number of apps available on the GooglePlay were around 1 million during this time\footnote{\url{http://en.wikipedia.org/wiki/Google_Play}}.
Furthermore, the average number of apps installed per user in our dataset is 42 with a standard deviation of 39.
Table \ref{tab:D} summarizes the main characteristics of our dataset $D$.

Figure~\ref{fig:cdf} depicts the cumulative distribution of the number of apps installed by a particular user.
We note that more than 90\% of users have 100 or fewer applications. 
Probability distribution of the number of users who installed a particular app is depicted by Figure~\ref{fig:pdf_users_cut}.
Notice that more than half of the apps are contained by only a single record in $D$.

\paragraph{Ethical Considerations}
The data were collected with the users' consent, and they were explicitly informed that their data could be used and shared for various research projects.
In fact, the Carat privacy policy (available at \url{http://carat.cs.helsinki.fi}) clearly specifies that ``Carat is a research project, so we reserve the right to publish our results online and in academic publications. 
We also reserve the right to release the data sets into the public domain."
Also, the dataset was shared with us by the Carat team in a pseudo-anonymised form. 
In particular, identifiers were removed, and each application name was replaced with its SHA1 hash. 
It contained 54,893 records~\cite{Truong:2014}, i.e. one record per user. 
Each record is composed of the list of applications installed by the user.
Furthermore, the data sharing agreement that we signed, stipulated that we cannot use the dataset to deanonymize the users in the dataset.

\begin{table}
\centering
\begin{tabular}{|l|l|}
\hline
Dataset size $|D|$ & $54,893$ \\ \hline 
\# of all apps in $D$ & $92,210$ \\ \hline
Maximum record size $\max_u|D_u|$ & 541\\ \hline
Minimum record size $\min_u|D_u|$ & 1 \\ \hline
Average record size & 42 \\ \hline
Std.dev of record size & 39 \\ \hline
\end{tabular}
\caption{\label{tab:D} Characteristics of our dataset $D$}
\end{table}

\subsection{Results}

We find that 98.93\% of users have unique set of installed apps in $D$, i.e., there does not exist any other user with the same set of installed apps.
This means that if we know the list of all the installed apps of a user in the dataset, we can identify that user in the dataset with a probability of 0.99.
As the adversary might not always be aware of all the installed apps of a user in practice, we measure the unicity of $K$-apps for different values of $K$ using our dataset $D$. 
%However, a real life adversary might not always be aware of all the installed apps.
%The adversary we consider knows a subset of installed apps $K$-apps of a user and with the help of this information, tries to identify the user in the dataset.
%The following unicity study would give an idea about the success rate of re-identification of a user by this adversary.

%As we discussed earlier, there are two ways to sample: biased and un-biased(uniform). 
%For each type of sampling of $K$-apps, we check if they are only present in a single user and as a result, this gives us the unicity probability $P$ for $K$-apps.
%As we did not check all the possible $K$-apps and only checked the number of samples required calculated as per the above formula for a particular confidence and error, the unicity probability $P$ might have $\epsilon$ error and we can say this with $\sigma$ confidence.

 \begin{figure}[h]
 	\includegraphics[width=\linewidth]{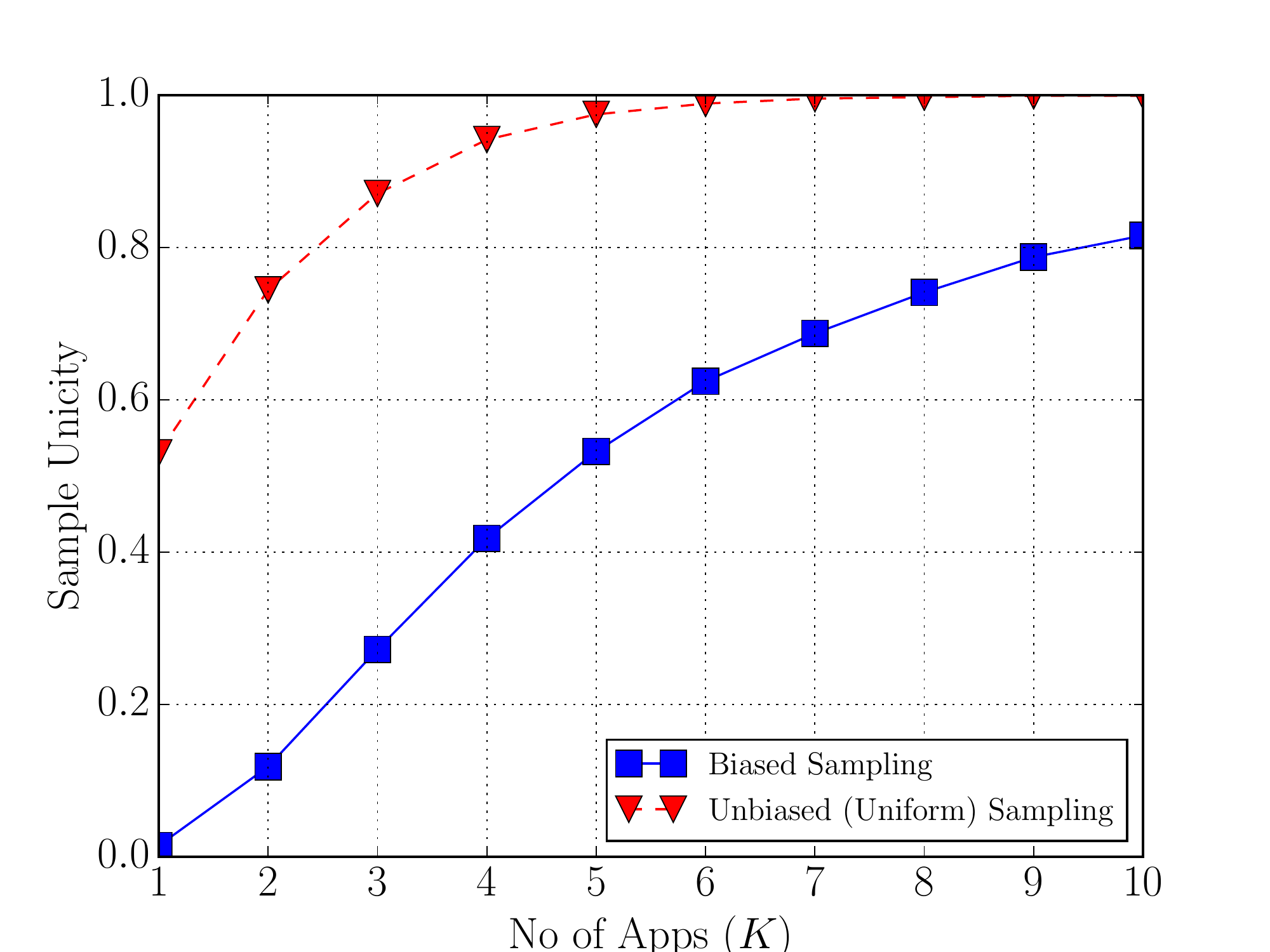}
 	\caption{Uniqueness probability as a function of $K$ for biased and unbiased sampling}
 	\label{fig:unicity_with_sampling_type}
 \end{figure}

Figure~\ref{fig:unicity_with_sampling_type} gives the unicity of $K$-apps with different values of $K$ (changing from 1 to 10) for the two different types of sampling techniques described in Section \ref{sec:sampling}: the biased sampling from \cite{Nature13, Science15} and our unbiased, uniform sampling described in Section \ref{sec:uniform_sampling}.
In each case, we computed the sample size using Inequality \ref{eq:chernoff} with maximum sampling error $\varepsilon = 0.01$ and confidence  $\sigma=0.99$.  Otherwise stated explicitly, we use this sample size in the sequel. This results in 26492 samples for each value of $K$. As biased sampling favours more popular $K$-apps, the sample unicity $\hat{H}_1$ is less than with our unbiased approach. In particular, the difference can be as large as 0.5 for smaller values of $K$, while it decreases as $K$ increases.
For the unbiased estimation, the sample unicity is 0.75 with $K=2$, and it reaches 0.99 when $K=6$.

Figure~\ref{fig:unicity_with_sampling_type} shows that the unicity of any $K$-apps is large and hence there would be a real privacy threat if such dataset was released.
Moreover, Figure \ref{fig:relative_abundance} depicts the relative abundance distribution in $D$, when $1 \leq K \leq 8$. RAD provides complementary information about users' privacy in $D$. In particular, even if the adversary cannot single out the record of the  target user in $D$, it might still learn new information about him/her. For example, if the known $K$-apps of the target user are shared by multiple users in $D$ and all these users have some identical apps besides the known $K$-apps, then the adversary learns that the target user also has these apps installed on his/her phone. This attack is often referred to as the homogeneity attack in the literature \cite{MachanavajjhalaKGV07}. We computed the required sample size using Inequality \ref{eq:rad_sample_size} with $\varepsilon = 0.01$ and $\sigma=0.99$ for $k=20$. This gives 41470 samples overall, which were taken with our uniform sampler $\mathcal{M}$.

\begin{figure}[h]
	\includegraphics[width=\linewidth]{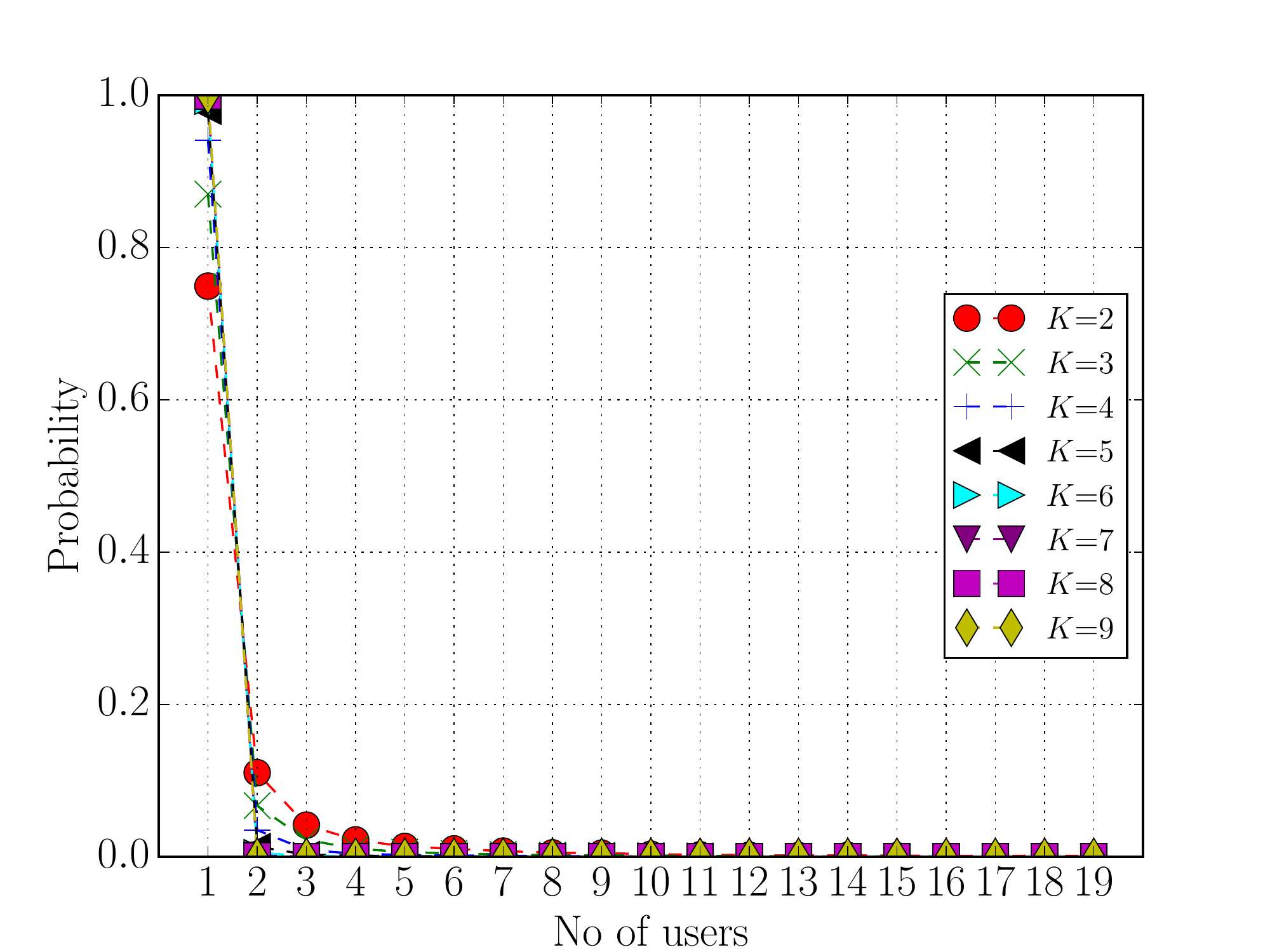}
	\caption{Relative abundance distribution of apps for different sizes of sets of apps}
	\label{fig:relative_abundance}
\end{figure}

%Figure~\ref{fig:relative_abundance} presents the number of users a sets of apps of different sizes $K$ are found in. 
%Here it is worthy to note that there is a big difference in the $K$-apps shared by only one user as compared to 2 or more users.

%As unicity is approximated  from our comparatively small dataset $D$, it is interesting to know how it would vary on larger datasets.
To study the effect of number of users on unicity, we randomly select subsets of users of different sizes from $D$, and calculate the sample unicity within these subsets. 
Figure~\ref{fig:unicity_changes_with_users} depicts how unicity changes with the number of users in our dataset.
We find that unicity decreases if the user number increases.
However, this decrease becomes less significant for larger number of users.
This is probably due to the fact that the number of apps starts to saturate if the user number increases.
\begin{figure}[!t]
        \centering
	\includegraphics[width=\linewidth]{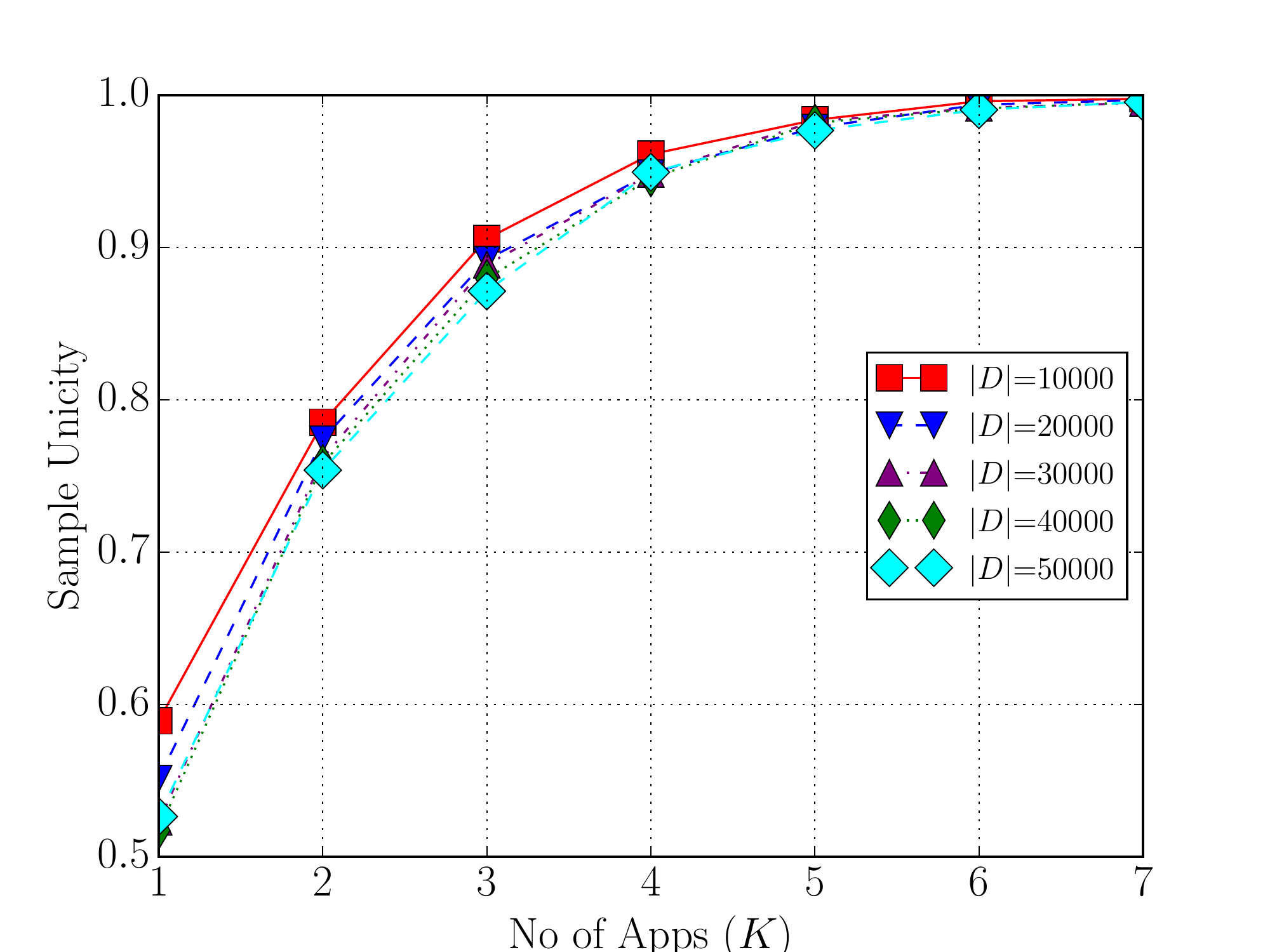}
	\caption{Effect of number of users}
	\label{fig:unicity_changes_with_users}
\end{figure}

As the size of our dataset is much less than the population size of all Android users worldwide (which was roughly 1 billion as of 2014 \footnote{\url{http://www.engadget.com/2014/06/25/google-io-2014-by-the-numbers/}} with 1.2 million different applications available on GooglePlay\footnote{\url{http://www.appbrain.com/stats/number-of-android-apps}}), we aim at predicting the unicity in a larger dataset (possibly in the whole population) in the next section.

%%%%%%%%%%%%%%%%%%%%%%%%%%%%%%%%%%%%%%%%%%%%%%%%%%%%%%%%
%=======================================================
%%%%%%%%%%%%%%%%%%%%%%%%%%%%%%%%%%%%%%%%%%%%%%%%%%%%%%%%
\section{Unicity Generalization for larger datasets}

\begin{figure*}[!t]
        \centering
        \begin{subfigure}[]{0.3\textwidth}
                \includegraphics[width=\textwidth]{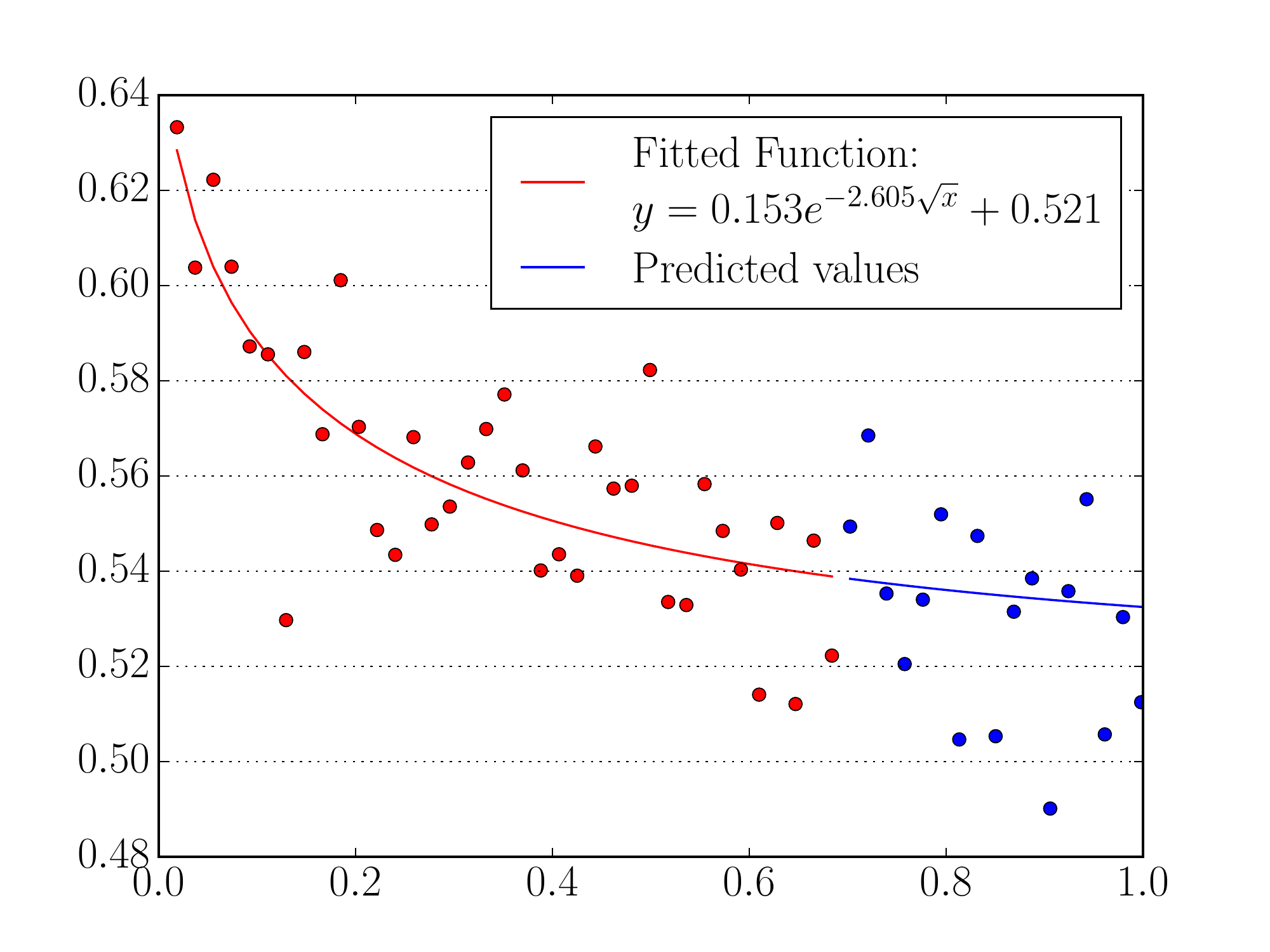}
                \caption{$K=1$, $\delta = 0.016$}
                \label{fig:K_1}
        \end{subfigure}
        \begin{subfigure}[]{0.3\textwidth}
                \includegraphics[width=\textwidth]{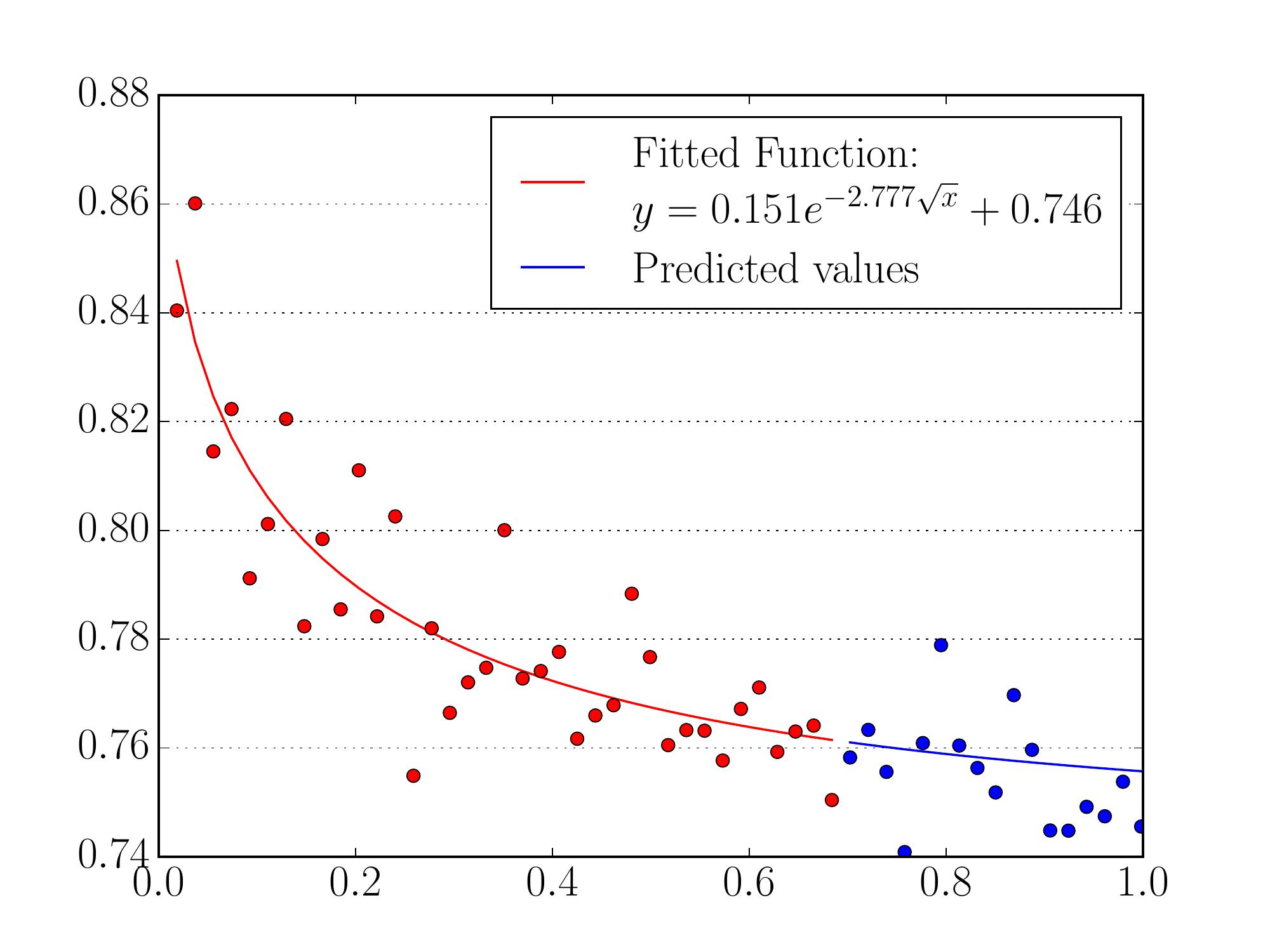}
                \caption{$K=2$, $\delta = 0.007$}
                \label{fig:K_2}
        \end{subfigure}
        \begin{subfigure}[]{0.3\textwidth}
                \includegraphics[width=\textwidth]{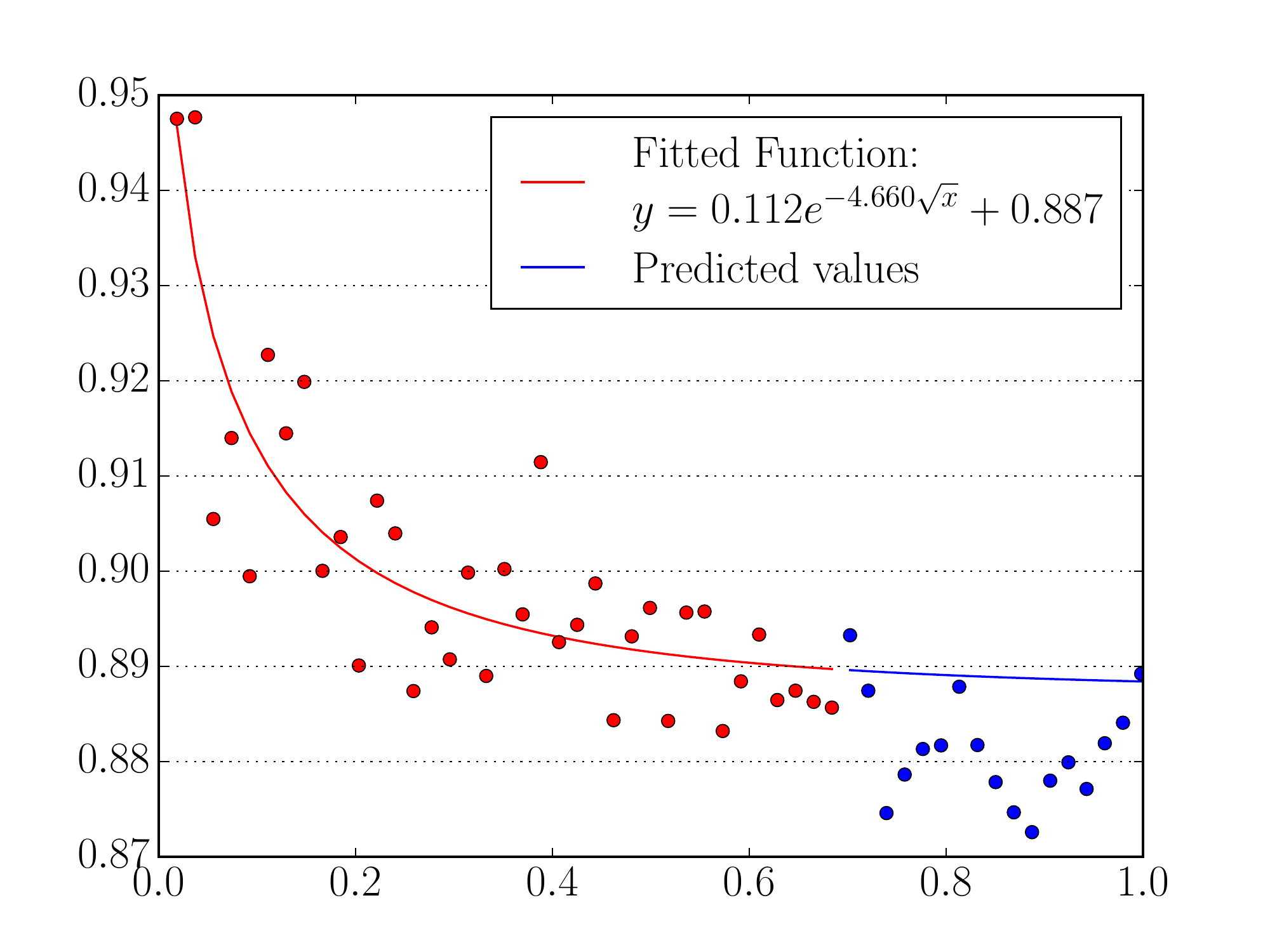}
                \caption{$K=3$, $\delta = 0.008$}
                \label{fig:K_3}
        \end{subfigure}
        \begin{subfigure}[]{0.3\textwidth}
                \includegraphics[width=\textwidth]{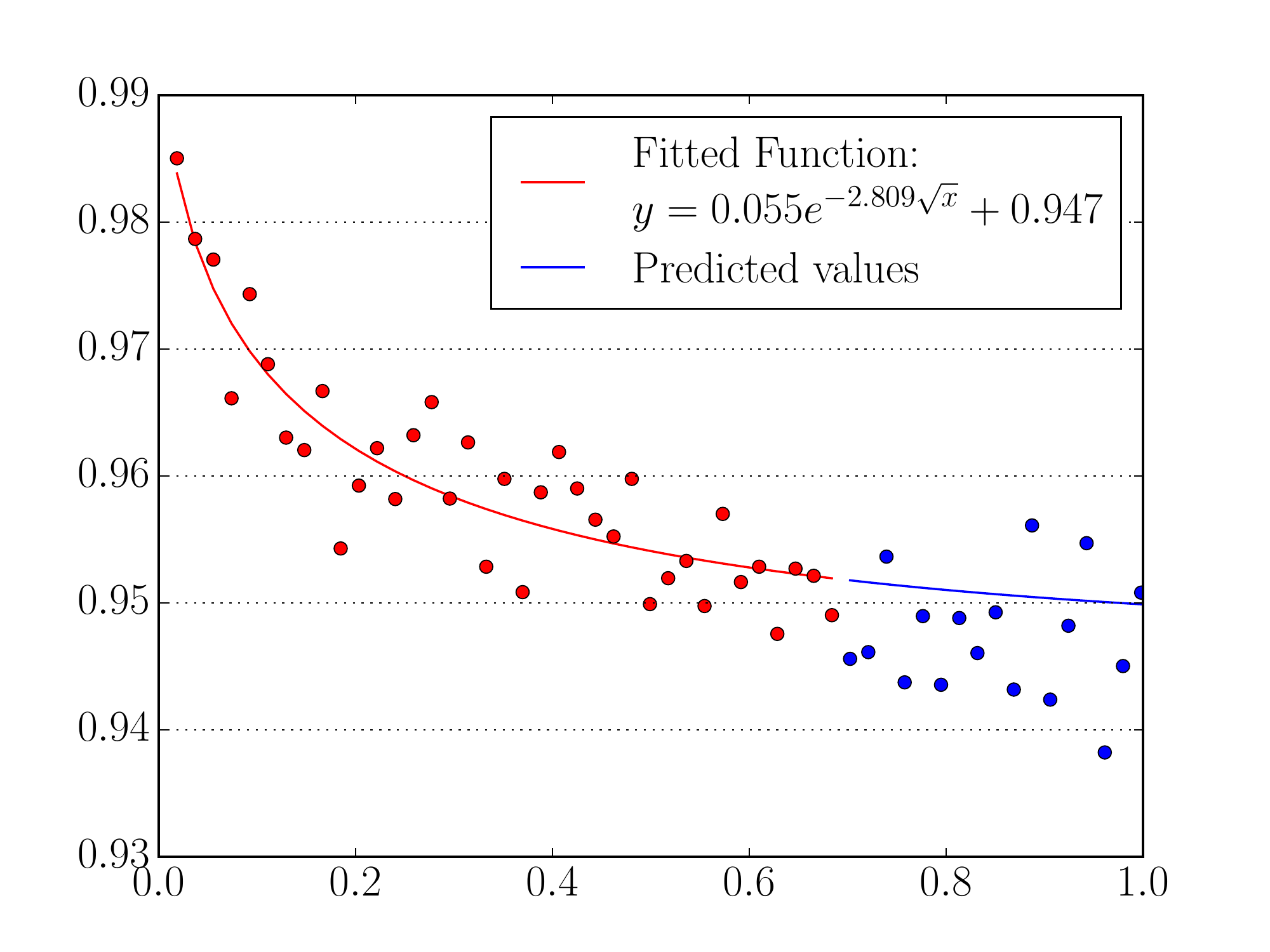}
                \caption{$K=4$, $\delta = 0.005$}
                \label{fig:K_4}
        \end{subfigure}
        \begin{subfigure}[]{0.3\textwidth}
                \includegraphics[width=\textwidth]{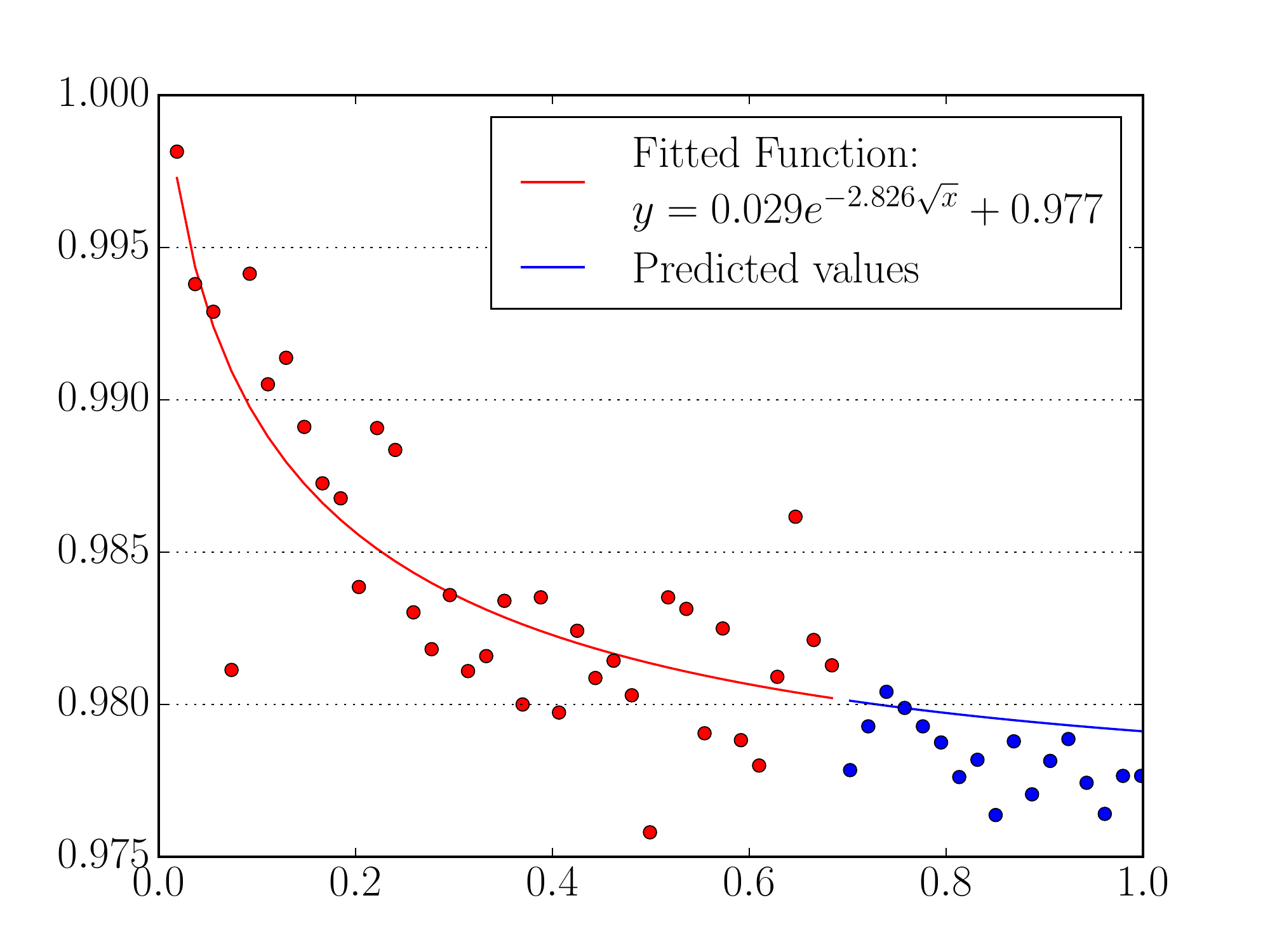}
                \caption{$K=5$, $\delta = 0.001$}
                \label{fig:K_5}
        \end{subfigure}
    \begin{subfigure}[]{0.3\textwidth}
                \includegraphics[width=\textwidth]{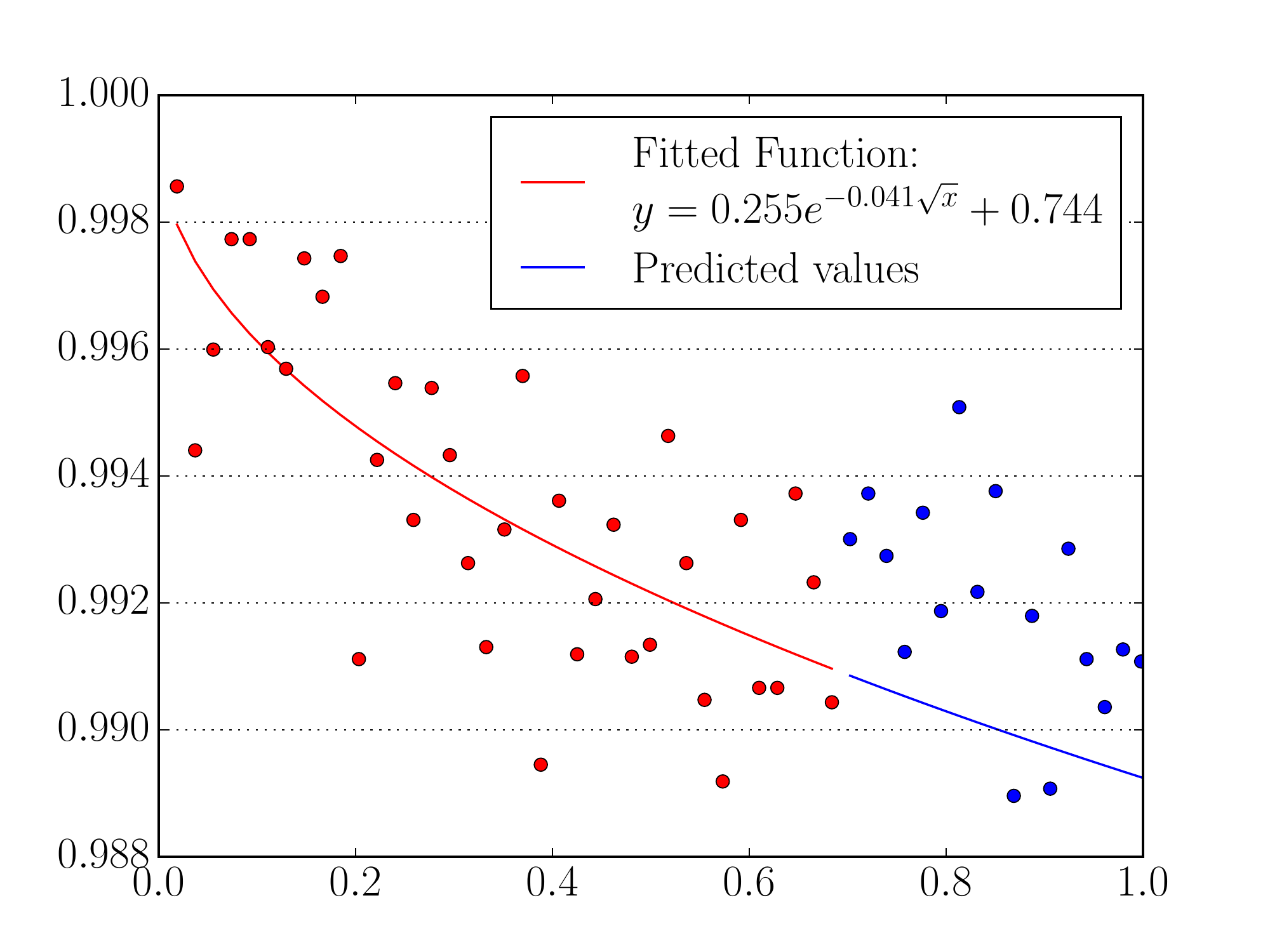}
                \caption{$K=6$, $\delta = 0.002$}
                \label{fig:K_6}
        \end{subfigure}
%    \begin{subfigure}[b]{0.3\textwidth}
%                \includegraphics[width=\textwidth]{non_linear_no_users_K_7_apps_square_root}
%                \caption{$K=7$, $\delta = 0.001$}
%                \label{fig:K_7}
%        \end{subfigure}
        \caption{Unicity generalization for different values of $K$, trained all with maximum $37000$ users. The learnt models (i.e., $f(x)$) are present in the legend. $x$-axis corresponds to normalized dataset sizes with a normalization factor of $1/54893$, and $y$-axis depicts sample unicity.}
    \label{fig:unicity_generalization}
\end{figure*}

Information surprisal can be used to measure uniqueness in the population $D$~\cite{Eckersley10}.
In our case, the population is all the Android users worldwide, to which we want to generalize our results.
As information surprisal of any $K$-apps $\{A_1, A_2, \ldots, A_K\}$ over $D$ is equal to $-\log(Pr[A_1,A_2,\ldots, A_K])$, we must need to first measure the co-occurence probability $Pr[A_1,A_2,\ldots, A_K]$ of these apps in $D$.
The co-occurence probability can be easily computed if we can assume that apps co-occur independently in the dataset as $Pr[A_1,A_2,\ldots, A_K] \approx \prod_{i=1}^K Pr[A_i]$, and $Pr[A_i]$ (the popularity of app $A_i$) can be obtained from the download count of $A_i$ available on Google PlayStore.
However, this is not the case in a real-world scenario as there exist correlation between apps installed by a user.
As our dataset is very likely to be too limited to capture this correlation (as our dataset contain only 93K distinct apps whereas there are more than 1.2 million available apps on GooglePlay), we cannot take this approach to measure the uniqueness in the population of Android users.
We rather employ regression analysis on our dataset which does not rely on this correlation information to predict the unicity in a larger dataset.

For regression analysis, we randomly create datasets of different sizes from our original dataset and compute the sample unicity for these datasets of different sizes.
This gives us the tuples ($x$,$y$) where $x$ is the number of users in a particular dataset (independent variable) and $y$ is the calculated unicity value dependent on $x$. 
Here, we assume that the unicity value $y$ only depends on the number of users $x$.  
We must note that, in reality, unicity depends on many factors such 
as the characteristics of the users, how many (un)popular applications users tend to have, etc.
As it is difficult to take into account all these factors either because they are unknown or hard to measure, we assume that unicity in general is a ``proper'' function of only the total number of users in the dataset. That is, all other dependent factors are implicitly incorporated into the model, i.e., the general form of the function.

Once we have these ($x$,$y$) tuples, our goal is to select the best model and its parameters that capture the relation between $x$ and $y$.
The overall approach is as follows: 
we divide our ($x$,$y$) tuples in training and test sets. 
We select the best model (i.e., a function family) based on the general characteristics of application unicity and then learn its exact parameters using our training set.
%$20\%$ of the tuples in the training set corresponding to larger values of $x$ were used for validation set. 
Finally, using the best model thus obtained, we test its accuracy on the test set.
This model should be able to predict the unicity value for any dataset of arbitrary size.

\paragraph{Training and Testing}
We divide our original dataset in 54 smaller datasets, each of size varying from 1k to 54k.
We take the first 70\%  of all ($x$,$y$) points for training and the last 30\%  (corresponding to larger datasets) for testing.
We deliberately take the last points corresponding to larger datasets for testing set because we aim to evaluate our model performance on larger datasets, i.e., we want to test how accurately the learned model could be extrapolated.

As we divided our datasets by randomly selecting users out of the original dataset, users in the training and testing set may overlap.
However, we found that unicity merely depends on the number of users in the dataset and not specifically on the underlying individuals.
For example, we computed the unicity of 50 different sets of 1000 users selected randomly, and found out that the variance of the measured sample unicity is very small.
%As we use only the total number of users for regression analysis, it does not matter who were the underlying users.

\paragraph{Model selection}

To select our model, we first tried linear regression with non-linear basis functions (polynomials of various orders) with and without regularization.
However, they provided very inaccurate predictions of unicity. Finally, we selected the following exponential model describing an exponential decay of unicity:
\begin{align}
\label{eq:model}
f(x) = a\cdot \exp(-b\sqrt{x}) + c
\end{align}

The rationale behind choosing this model is as follows. Figure~\ref{fig:correlation_users_apps} shows that if additional users were added to our dataset, the number of apps would reach the maximum number of apps in the population early as there are fewer apps on GooglePlay than total number of Android users.
This suggests that, after a certain point, additional users would not bring many new apps but still, they would bring new combinations of already existing apps.
The addition of new combinations of apps should lead to the increase in unicity. However, the newly added users can lead to the decrease in unicity as well due to the fact that they can also have many already existing combinations of apps.
%As the effects of introducing new users and the addition of new combinations of apps run opposite to each other, we suppose that unicity converges to a value greater than zero which is denoted by $c$ in Equation~\ref{eq:model}.
As these two effects of adding new users to the dataset run opposite to each other, we suppose that unicity converges to a value greater than zero which is denoted by $c$ in Equation~\ref{eq:model}.
Indeed, as Figure~\ref{fig:unicity_changes_with_users} shows, although unicity decreases with the increase in the user number, the amount of this decrease tends to decrease as well.
A similar observation was made in \cite{Nature13}.
Also, we used square root of $x$ in the exponent in Equation~\ref{eq:model} because taking square root is variance-stabilizing\footnote{\url{https://en.wikipedia.org/wiki/Variance-stabilizing_transformation}}.
In fact, we tried other powers of $x$ in the exponent but square root lead to the best results.

The goal of the regression is to compute parameters $a$,$b$ and $c$ in Equation~\ref{eq:model} from the training set  ($x$,$y$) tuples.
In fact, these parameters might be computed employing either standard non-linear regression directly or by first transforming Equation~\ref{eq:model} into linear form and then applying linear regression.
We use standard non-linear regression because it explicitly computes the lower bound on the unicity value (i.e., $c$ in Formula \ref{eq:model}).
%This is a standard non-linear least square fitting problem which can be approximated with any numeric minimization method such as the Levenberg-Marquardt algorithm \cite{Levenberg44am}.
The value of $x$ is normalized\footnote{\url{https://en.wikipedia.org/wiki/Feature_scaling}} by dividing $x$ with the maximum size of the dataset for which we want to predict the unicity value.
%Even though we used for normalization the maximum size of the dataset we wanted to predict the unicity value, it can be anything as long as $0 < x < 1$.
%This is because $x$ is in square root and square root of values between 0 to 1 is more than the value itself while the square root of values greater than 1 is less than the value itself.

\paragraph{Results}

As an error metric, we measure the average absolute error denoted by $\delta$, i.e., 
$$
\delta = (1/n) \sum_{i}^{n} \left| y_i - f(x_i)\right|
$$
where $n$ is the number of predicted points,  $y_i$ is the real unicity value, and $f(x_i)$ is the predicted value.

Figure~\ref{fig:unicity_generalization} presents how our exponential model performs on the test set for different values of $K$.
Although our model can predict the trend of the unicity for large number of users, it slightly overestimates the real unicity in the test set.
Nevertheless, the average error $\delta$ on the test set is only around 0.01. 
As our app dataset is very small as compared to the whole Android population, we cannot evaluate performance of our model for large number of users, e.g., a few million, or the whole Android population.
Therefore, we cannot claim that our model will be able to accurately predict the unicity for datasets having large number of users even if it performs reasonably well on our test data.

%The predicted unicity values (i.e., $c$ in Formula \ref{eq:model}) in the whole Android population is roughly the same as in our dataset $D$ (Figure \ref{fig:unicity_with_sampling_type}). 
%they suggest that our model is not capable of predicting the uniqueness in the whole population using such a limited number of users.  

%Unicity (K = 2, N = 26492): 0.749282802355
%Unicity (K = 3, N = 26492): 0.870073984599
%Unicity (K = 4, N = 26492): 0.940774573456
%Unicity (K = 5, N = 26492): 0.978974784841
%Unicity (K = 6, N = 26492): 0.992224067643
%Unicity (K = 7, N = 26492): 0.997810659822
%Unicity (K = 8, N = 26492): 0.999094066133
%Unicity (K = 9, N = 26492): 0.999849011022

%Finally, we tried different powers of independent variable $x$ in the exponent of $e$ and found out that square root of independent variable in the exponent fits very well the model.

\begin{figure}[!t]
        \centering
	\includegraphics[width=0.7\linewidth]{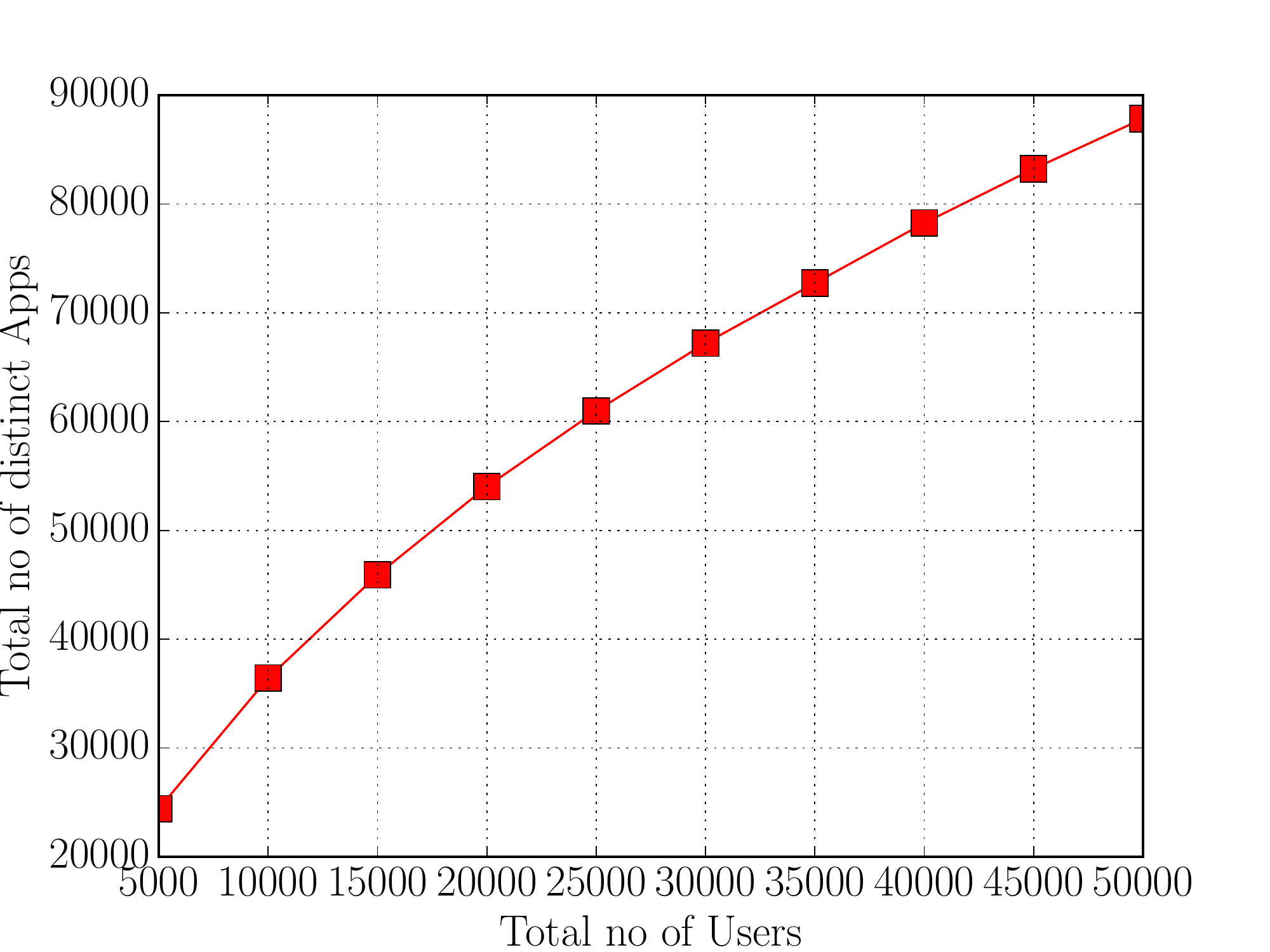}
	\caption{No of distinct apps installed by users}
	\label{fig:correlation_users_apps}
\end{figure}

%\begin{figure*}
%        \centering
%        \begin{subfigure}[b]{0.3\textwidth}
%                \includegraphics[width=\textwidth]{non_linear_no_users_K_1_apps_square_root.pdf}
%                \caption{A gull}
%                \label{fig:gull}
%        \end{subfigure}
%        ~ %add desired spacing between images, e. g. ~, \quad, \qquad, \hfill etc.
%          %(or a blank line to force the subfigure onto a new line)
%        \begin{subfigure}[b]{0.3\textwidth}
%                \includegraphics[width=\textwidth]{non_linear_no_users_K_2_apps_square_root.pdf}
%                \caption{A tiger}
%                \label{fig:tiger}
%        \end{subfigure}
%        ~ %add desired spacing between images, e. g. ~, \quad, \qquad, \hfill etc.
%          %(or a blank line to force the subfigure onto a new line)
%        \begin{subfigure}[b]{0.3\textwidth}
%                \includegraphics[width=\textwidth]{non_linear_no_users_K_3_apps_square_root.pdf}
%                \caption{A mouse}
%                \label{fig:mouse}
%        \end{subfigure}
%        \caption{Pictures of animals}\label{fig:animals}
%\end{figure*}

\begin{figure*}[!t]
        \centering
        \begin{subfigure}[]{0.3\textwidth}
                \includegraphics[width=\textwidth]{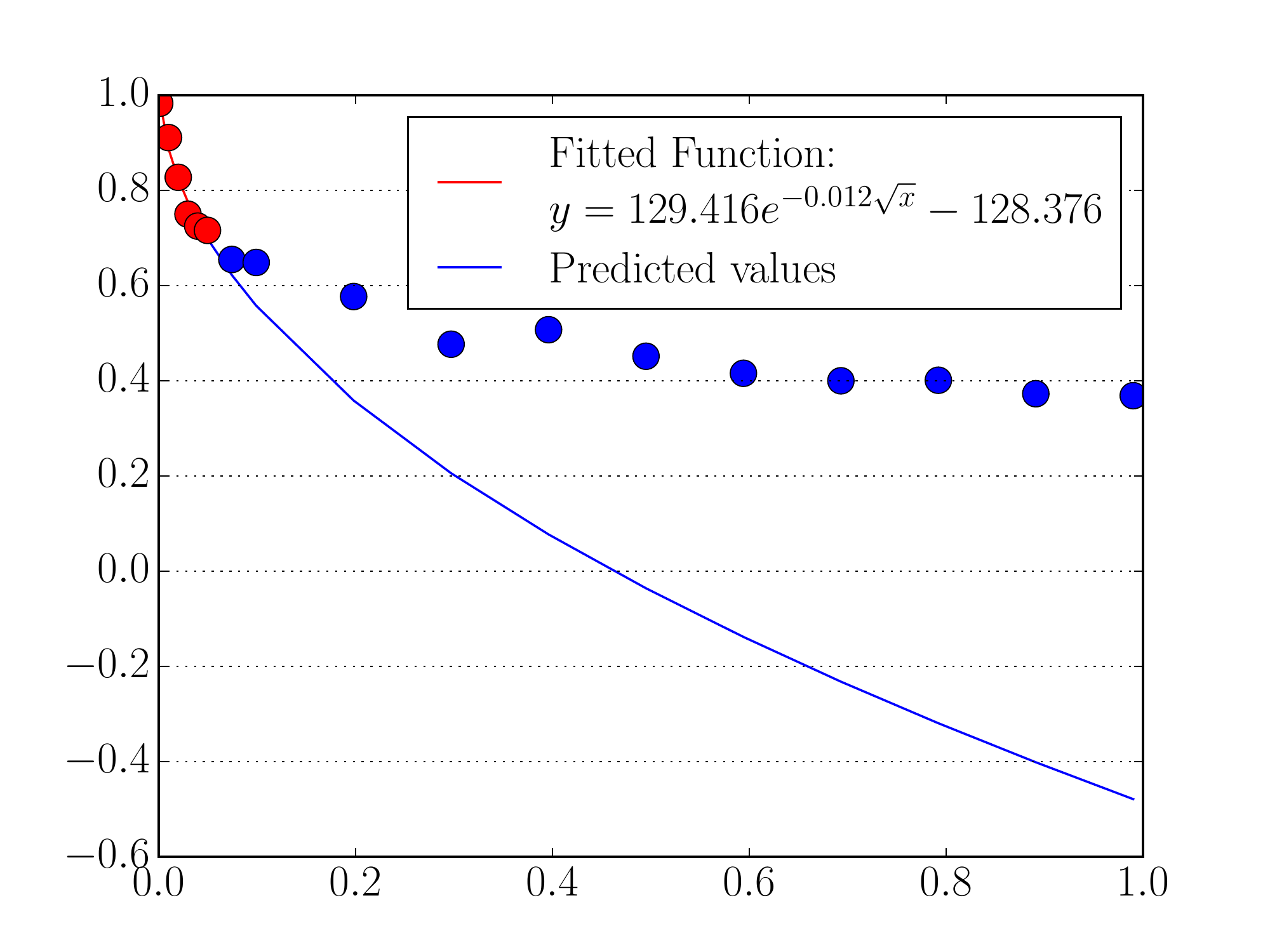}
                \caption{$K=4$, trained with max $|D|=50000$, $\delta = 0.460$}
                \label{fig:K_4_loc_50000}
        \end{subfigure}
    \begin{subfigure}[]{0.3\textwidth}
                \includegraphics[width=\textwidth]{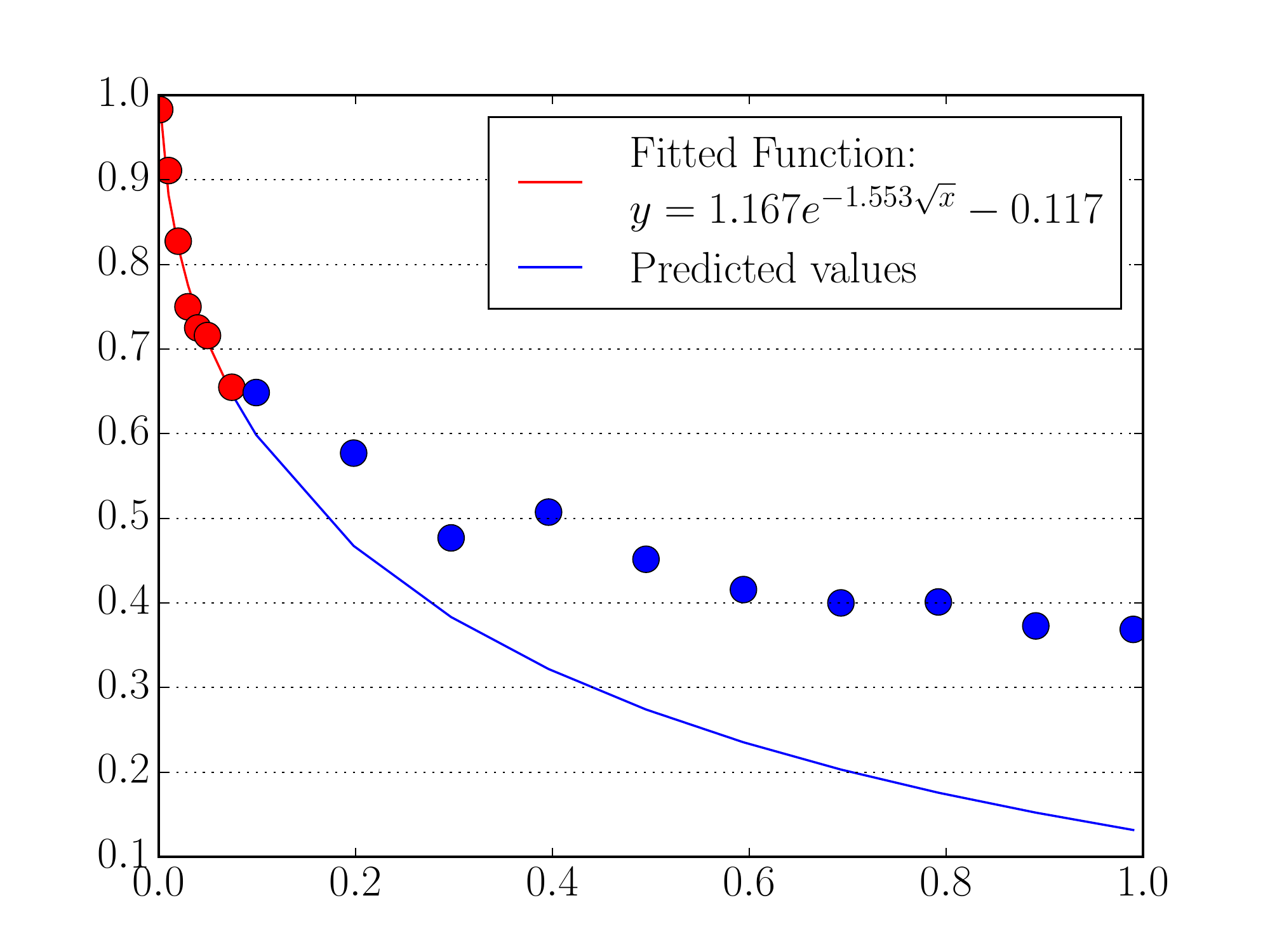}
                \caption{$K=4$, trained with max $|D|=75000$, $\delta = 0.168$}
                \label{fig:K_4_loc_75000}
        \end{subfigure}   
    \begin{subfigure}[]{0.3\textwidth}
                \includegraphics[width=\textwidth]{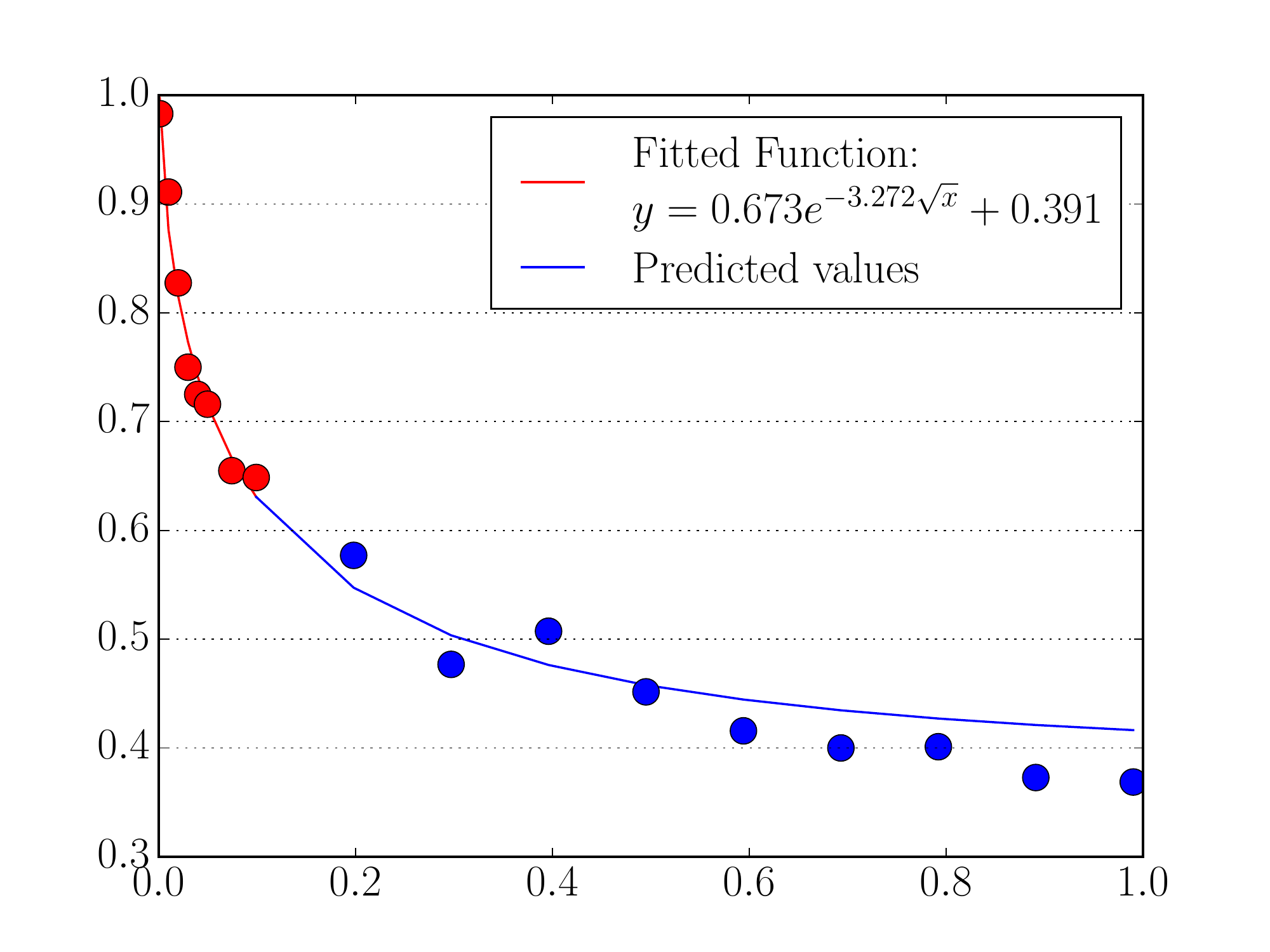}
                \caption{$K=4$, trained with max $|D|=100000$, $\delta = 0.031$}
                \label{fig:K_4_loc_100000}
        \end{subfigure}
    \begin{subfigure}[]{0.3\textwidth}
                \includegraphics[width=\textwidth]{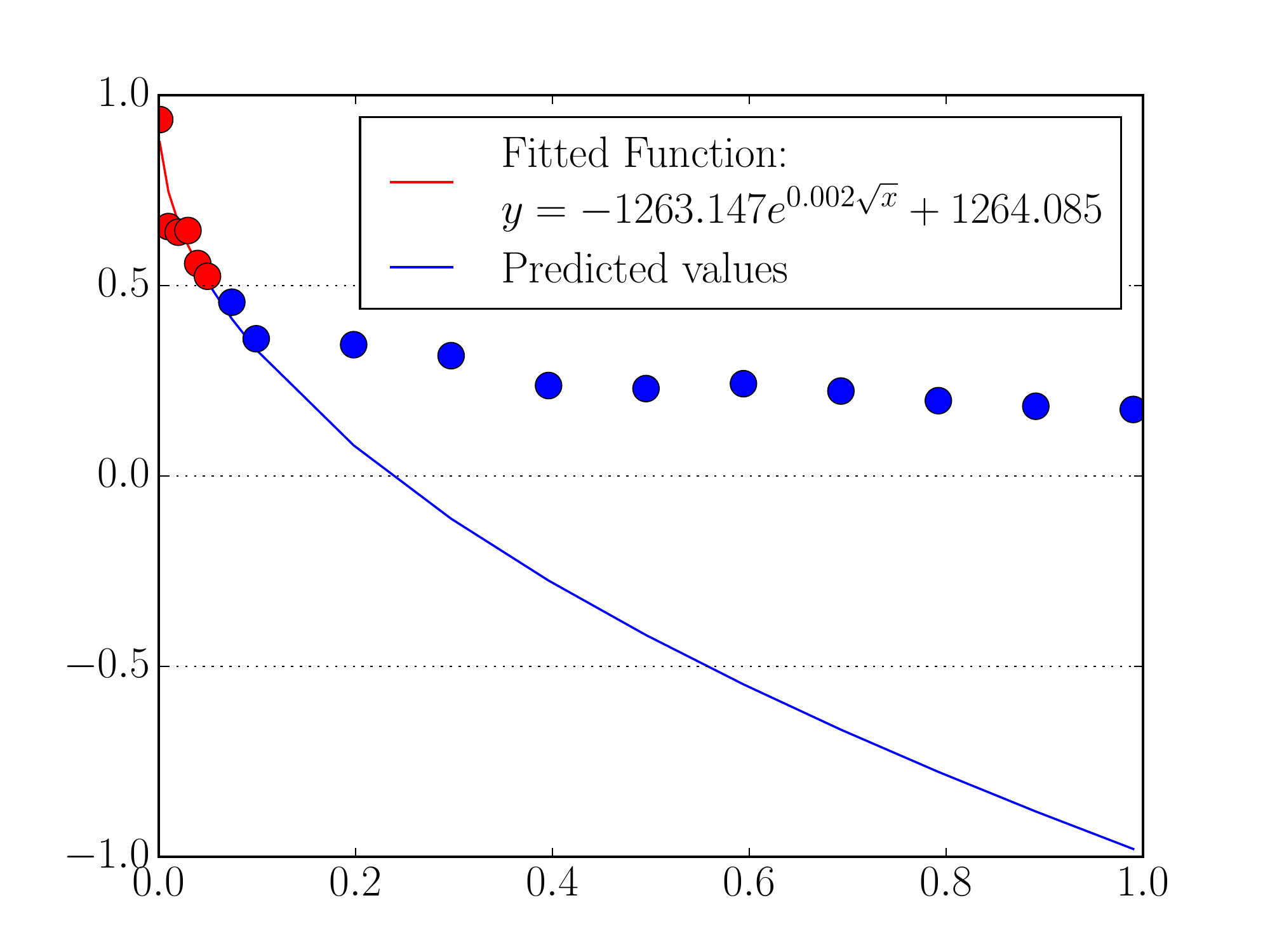}
                \caption{$K=3$, trained with max $|D|=50000$, $\delta = 0.618$}
                \label{fig:K_3_loc_50000}
        \end{subfigure}
    \begin{subfigure}[]{0.3\textwidth}
                \includegraphics[width=\textwidth]{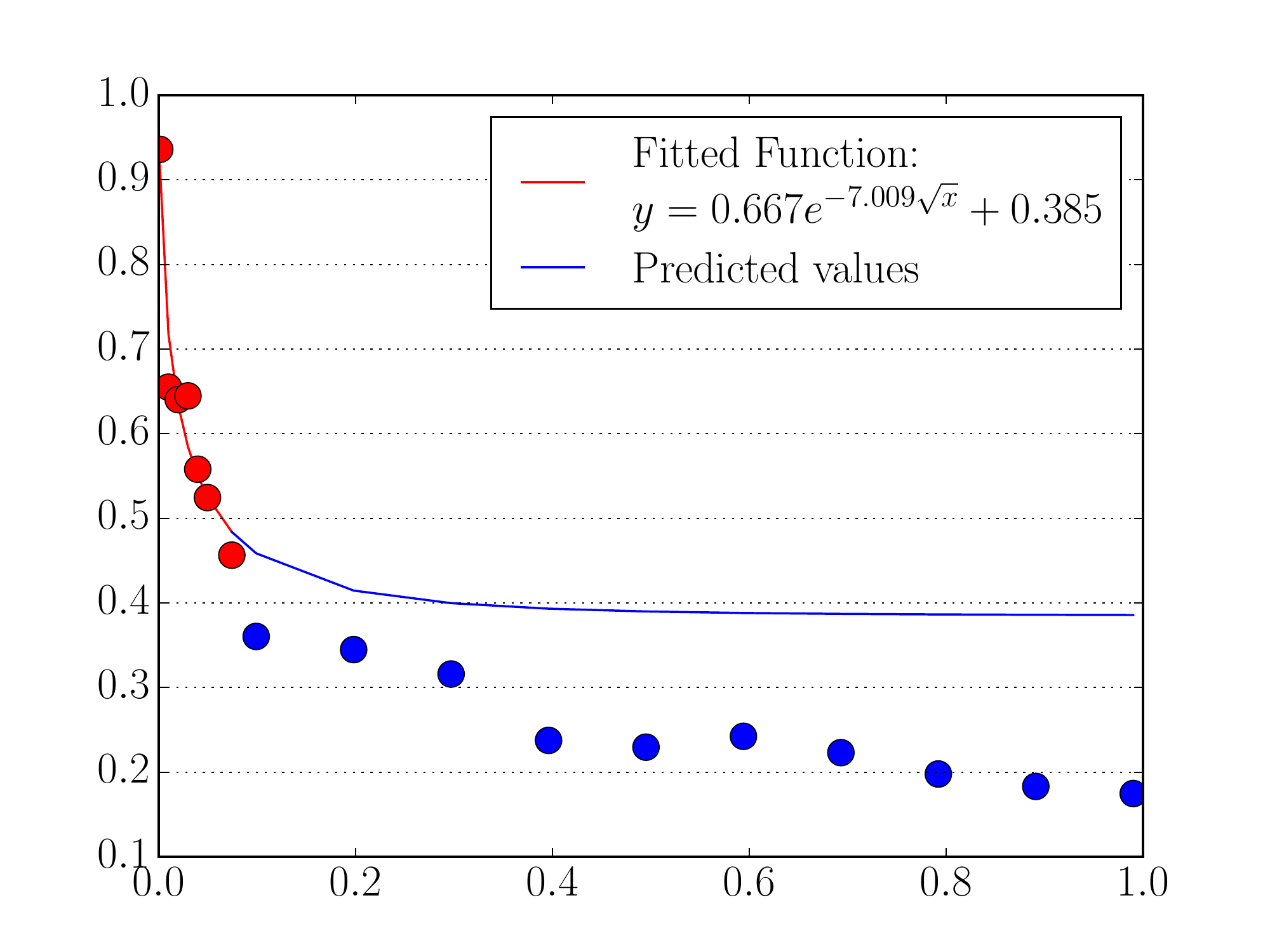}
                \caption{$K=3$, trained with max $|D|=75000$, $\delta = 0.148$}
                \label{fig:K_3_loc_75000}
        \end{subfigure}
    \begin{subfigure}[]{0.3\textwidth}
                \includegraphics[width=\textwidth]{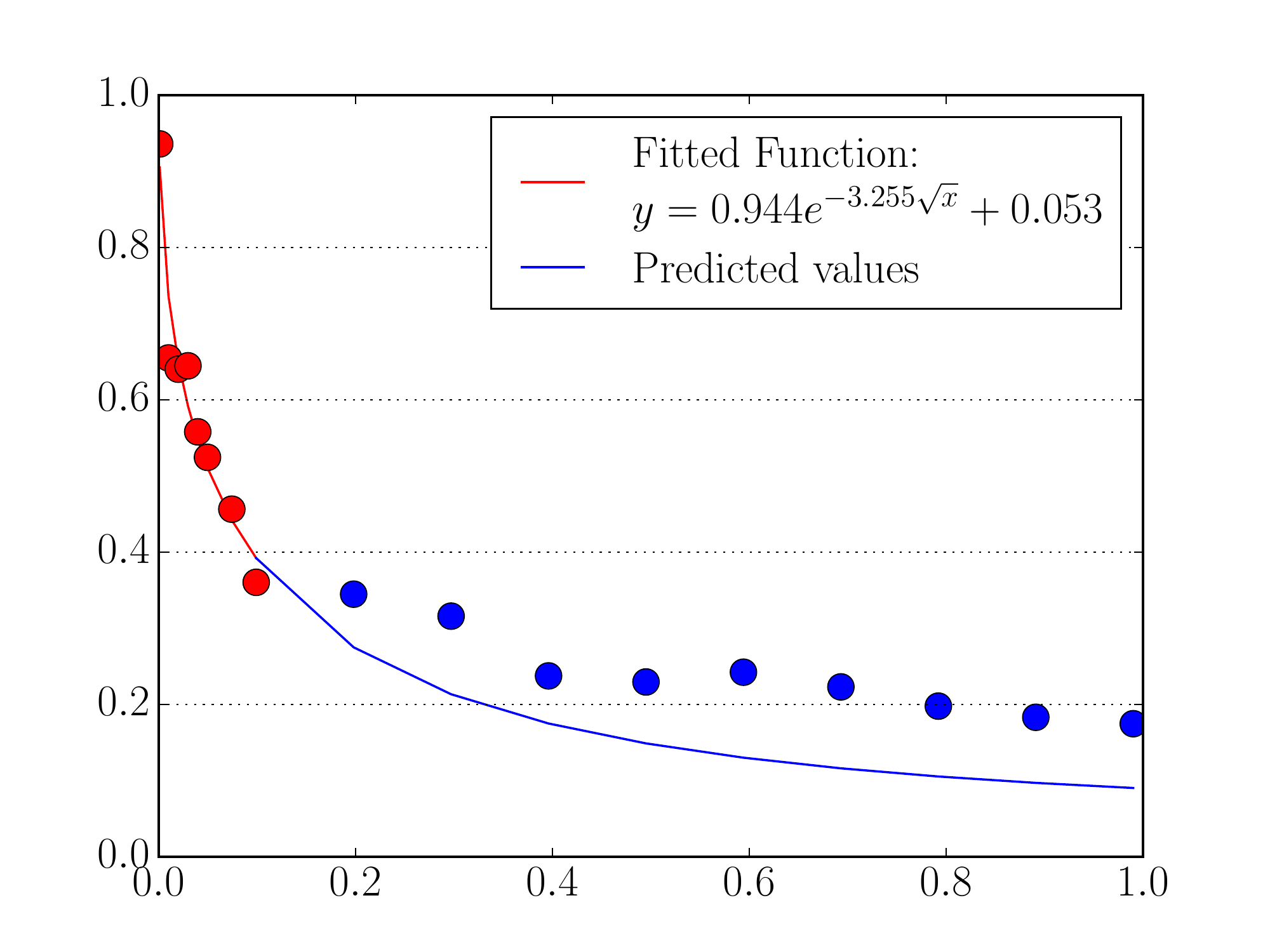}
                \caption{$K=3$, trained with max $|D|=100000$, $\delta = 0.089$}
                \label{fig:K_3_loc_100000}
        \end{subfigure}
    \begin{subfigure}[]{0.3\textwidth}
                \includegraphics[width=\textwidth]{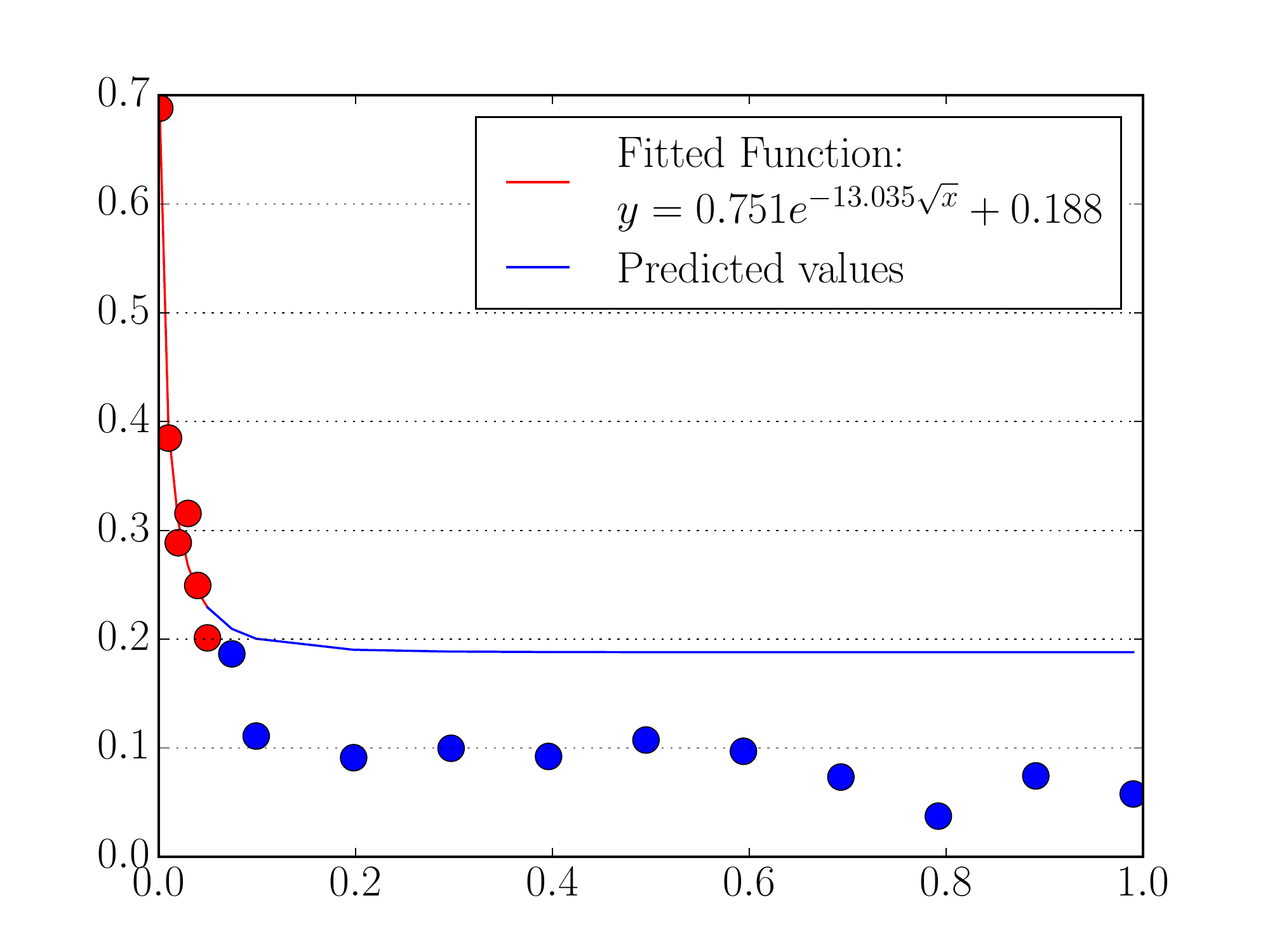}
                \caption{$K=2$, trained with max $|D|=50000$, $\delta = 0.098$}
                \label{fig:K_2_loc_50000}
        \end{subfigure}
    \begin{subfigure}[]{0.3\textwidth}
                \includegraphics[width=\textwidth]{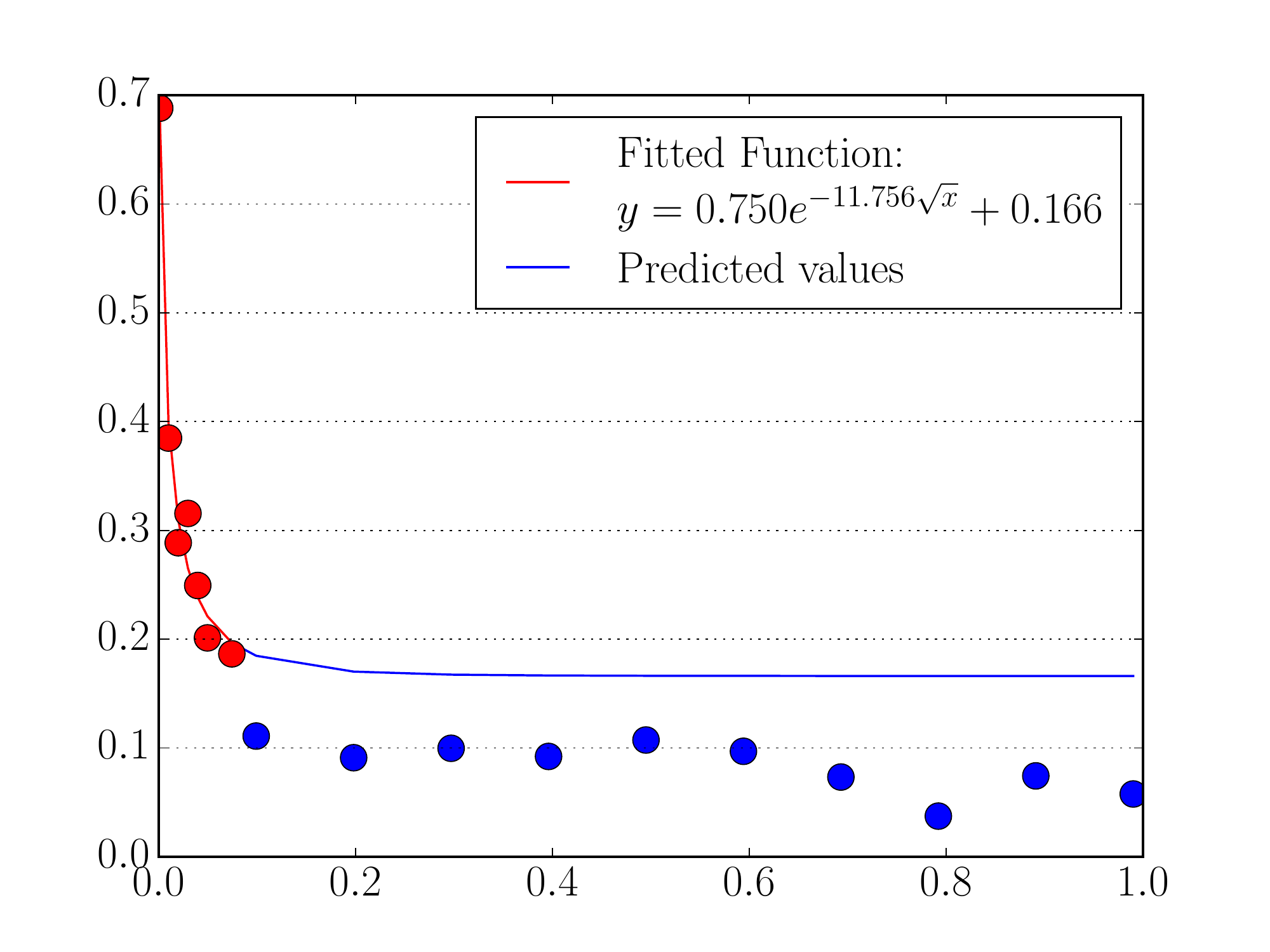}
                \caption{$K=2$, trained with max $|D|=75000$, $\delta = 0.085$}
                \label{fig:K_2_loc_75000}
        \end{subfigure}
    \begin{subfigure}[]{0.3\textwidth}
                \includegraphics[width=\textwidth]{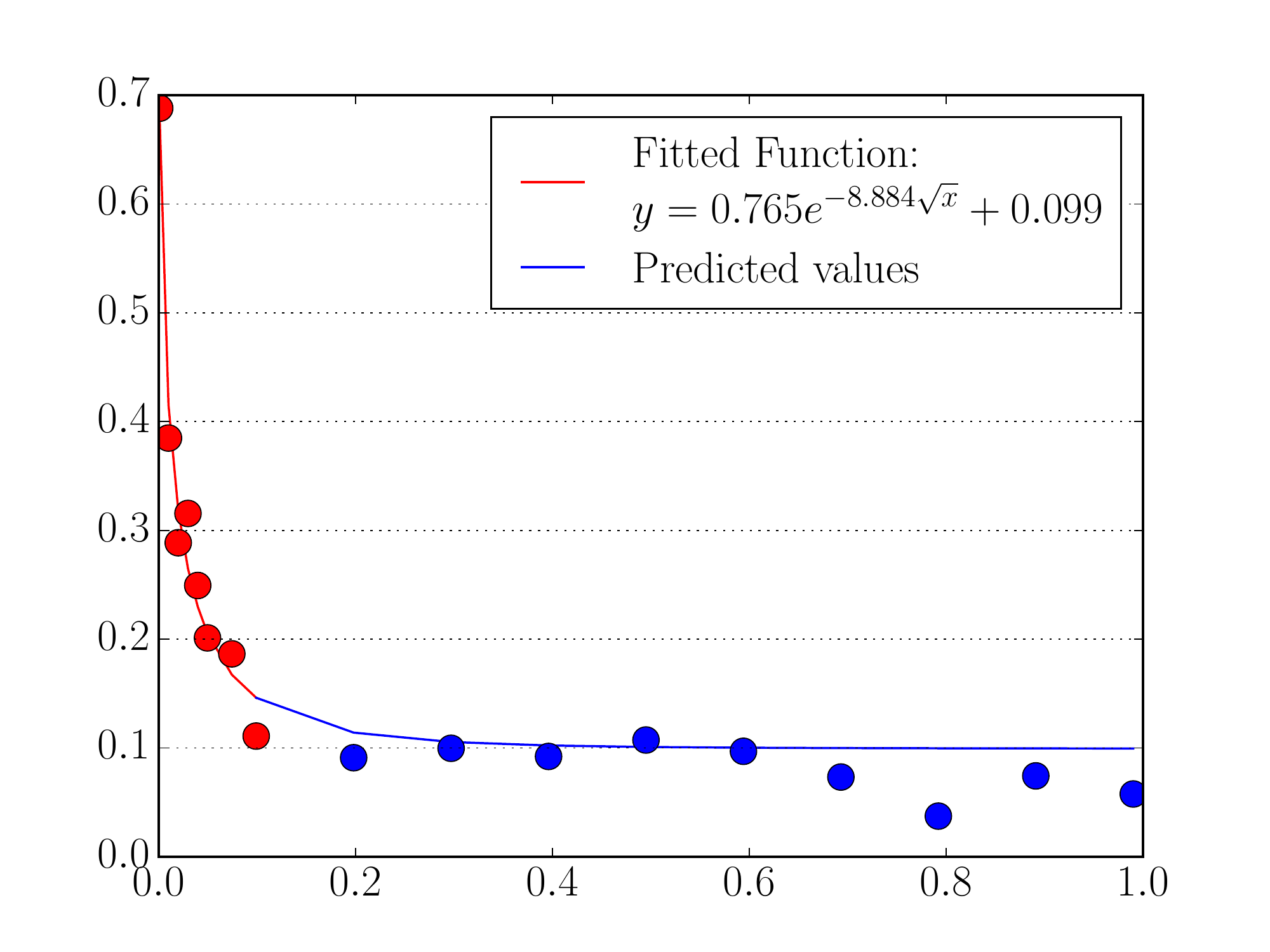}
                \caption{$K=2$, trained with max $|D|=100000$, $\delta = 0.023$}
                \label{fig:K_2_loc_100000}
        \end{subfigure}
        \caption{Unicity generalization for different values of $K$ for location data, trained with datasets of different sizes. The learnt models (i.e., $f(x)$) are present in the legend. $x$-axis corresponds to normalized dataset sizes with a normalization factor of $1/10^6$, and $y$-axis shows sample unicity.}\label{fig:unicity_location_generalization}
\end{figure*}

\paragraph{Model validation on a different dataset}
To further demonstrate that our model is a meaningful approach to predict unicity in large populations, we test it on a large mobility dataset provided by a telecom operator in Europe.
This dataset contains the Call Data Records (CDR) of 1 million users from a large european city over 6 weeks. 
Each record in the dataset corresponds to a user and contains the set of his/her visited cell towers, where the total number of different cell towers is 1303. 
From this dataset, we created smaller datasets of different sizes ($x$ ranging from 1000 to 1 million users).
% should we put below real data points???
%1000, 10000, 20000, 30000, 40000, 50000, 75000, 100000, 200000, 300000, 400000, 500000, 600000, 700000, 800000, 900000 and 1000000.
Then, we trained our model on the first 6 points (i.e., until the dataset size of $50,000$).  
Figure \ref{fig:unicity_location_generalization} shows that the model does not predict accurately the unicity for large datasets and the error can be as large as 0.6.
Next, we trained the model on the first 7 points (i.e., until the dataset size of $75,000$).
In this case, we find that the model performs significantly better than in the previous case with an error of $0.13$ on average.
Finally, we trained the model on the first 8 points (i.e., until the dataset size of $100,000$). 
In this case, we find that the model have accurate predictions for larger datasets, e.g., for a test dataset of size 1 million (10 times more than the maximum size of the dataset used in the training phase), and the error is $0.05$ on average.

We find that a mobility dataset of 0.1 million users is sufficient to learn an accurate model and predict the unicity values for a larger mobility dataset.
However, as we saw earlier, the model is not able to predict well the unicity of a large population if it is trained on a dataset of only $50,000$ users.
This may suggest that $50,000$ users might be too small in general to learn an accurate model and therefore, our app dataset of $50,000$ users might not be sufficient to learn the model.
On the other hand, even if this model performs well on a mobility dataset with 1 million users, it does not necessarily imply its good performance on large application datasets due to the different data and user characteristics.
Nevertheless, these two experiments together show that our exponential model can be a meaningful approach to predict unicity in large populations.

\section{Conclusion}
The paper shows that the list of installed applications is quite unique.
%that a user installed on
%her smartphone is, in most of the
%cases, unique and can be used as an identifier.
This result has few implications on user's privacy. 
First, since this metadata is unique, it could easily be used to profile users, e.g., based on the category of installed apps. 
This is what Twitter is doing to provide interest-based targeted ads to users.
Second, as a combination of even small number of installed apps is quite unique, this information could be used to re-identify users in a dataset. 
For example, if Twitter decided to publish the list of apps installed by its users on their smartphones, it would be easy for anyone, who knows 4 or 5 apps of a given user, to re-identify him and discover other apps that are also installed on his smartphone. 
This makes anonymization of this information challenging, and this is part of our future work.

%People are unique in the way they are and they behave. 
%This is why it is difficult to protect their privacy. 
In general, mobile users reveal many pieces of information that, when combined together, provide a lot of information about users and can be used to build personalized profiles. 
%For example, by combining the locations and the list of apps of a user, a service, such as Twitter, can uniquely categorize a user. 
Since people are unique in many different known and unknown ways, preserving the privacy of mobile users is very challenging. %at least, from technical point of view. 
New protection measures need to be devised.

%This dataset should not be published as it is if someone was to publish it because it carries privacy risks.
%If an adversary knows a subset of apps, it is possible to identify some users and hence, the adversay can learn about other apps installed by a user.
%This could lead to a privacy breach.
%As we have shown in this paper that unicity of apps is quite high, there is a high probability that the adversary would be able to re-identify a big percentage of the users.

%We can annonymise these dataset if we need to publish this dataset.
%For anonymisation, we can use $k^m$ anonymity model~\cite{Terrovitis:2008}.
%What is this model? Is this model sufficient? State of the art techniques of $k^m$ anonymity model that could be used are~\cite{Terrovitis:2008},~\cite{He:2009} and~\cite{Terrovitis:2011}.
%If not enough, differential privacy model could be used.
%Concerning the works that could be used to anonymise our data...imho,
%these works could be used:
%1. all works that talk about how to publish set-valued data form VLDB 08
%and VLDB 09 using k\^m annonymity model
%2. The work from Chen to publish set-valued data using difference
%privacy annonymity model~\cite{Chen:2011}

\section*{Acknowledgements}
We thank the whole Carat team for sharing their dataset with us. 
Also, our special thanks go to H.~Truong, N.~Asokan and S.~Tarkoma for discussions related to the dataset. 
This work is partially funded by Inria project lab CAPPRIS.

\bibliographystyle{abbrv}
\bibliography{paper_arxiv}

\appendix

\section{Proof of Theorem \ref{THM:MARKOV}}
\label{sec:app}

\paragraph{Metropolis-Hastings algorithm} Consider an ergodic  Markov  chain $\mc{G}$ with transition graph $G(\Lambda, E)$ and transition matrix $P_{\mc{G}}$, where $\Lambda$ is the finite state space.
 For each $x\in \Lambda$, 
let $\kappa: \Lambda \times \Lambda \rightarrow [0,1]$ denote a (not necessarily symmetric) proposal probability distribution function such that for all $x\in \Lambda$, $\kappa(x,x)  + \sum_{y\neq x} \kappa(x,y) = 1$.
%with the following properties:
%\begin{itemize}
%\item $\kappa(x,y) = \kappa(y,x) > 0$ for all $x,y$ such that $(x, y) \in E$, and $\kappa(x,y) = 0$ otherwise 
%\item For all $x\in \Omega$, $\kappa(x,x)  + \sum_{y\neq x} \kappa(x,y) = 1$
%\end{itemize}
The transitions of $\mc{G}$ are defined according to the \emph{Metropolis-Hastings rule} as follows. From any state $x \in \Lambda$, first select an $y\in \Lambda$ such that $(x,y) \in E$ with probability $\kappa(x,y)$. Then,
``accept" the transition from $x$ to $y$ with probability $\min\left(1, \frac{\pi(y)\kappa(y,x)}{\pi(x)\kappa(x,y)}\right)$, otherwise stay at $x$. 
%%Beri
%Since $\mc{M}$ is ergodic, its stationary distribution $\pi$ is unique, if it also satisfies the detailed balance condition, i.e., $\pi(x)P_{\mc{M}}(x,y) = \pi(y)P_{\mc{M}}(y,x)$. Therefore,  after sufficiently many transitions, the distribution of states will be very close to $\pi$.
It is not difficult to show \cite{LevinPW09mixing} that such a Markov chain $\mc{G}$ is reversible, i.e., $\pi(x)P_{\mc{G}}(x,y) = \pi(y)P_{\mc{G}}(y,x)$ and therefore its stationary distribution $\pi$ is unique \cite{LevinPW09mixing}. Consequently,  after sufficiently many transitions, the distribution of states will be very close to $\pi$. Notice that there is no need to compute the normalization constant of $\pi$,  even if $|\Omega|$ is very large (i.e., an exponential function of $K$ in our problem), because it appears both in the numerator and denominator of the transition probabilities. 

%Notice that $\mc{G}$ provides a powerful technique to sample from any distribution $\pi$ over $\Omega$ as there is no need to compute the normalization constant of $\pi$,  even if $|\Omega|$ is very large (i.e., an exponential function of $K$ in our problem), because it appears both in the numerator and denominator of the transition probabilities. 

\begin{proof}[of Theorem \ref{THM:MARKOV}]
In each iteration, $\mathcal{M}$ can select any individual $u$ in $D$.  Hence, at any state, $\mathcal{M}$ can visit any state in $\Omega^K$. Therefore, $\mathcal{M}$ is connected and aperiodic. Also notice that 
\begin{align*}
\frac{\pi(C) \kappa(C,S)}{\pi(S) \kappa(S,C)} &= \frac{\kappa(C,S)}{\kappa(S,C)} \\
&= \frac{\sum_{\forall u : U_u \supseteq S}  1 / {|U_u| \choose K}}{ \sum_{\forall u : U_u \supseteq C}  1 / {|U_u| \choose K}}  \\
&=\frac{\sum_{\forall u : U_u \supseteq S} K!/|U|  \prod_{i = 1}^{K} \frac{1}{|U_u| - K + i}}{ \sum_{\forall u : U_u \supseteq C} K!/|U|  \prod_{i = 1}^{K} \frac{1}{|U_u| - K + i}} \\
&=  q(S)/q(C)
\end{align*}
where $\pi$ is the uniform distribution over $\Omega^K$. Therefore, a candidate next state is accepted with probability $\min\left(1, \frac{\pi(C)\kappa(C,S)}{\pi(S)\kappa(S,C)}\right) = \min\left(1, q(S)/q(C)\right)$, which means that
$\mathcal{M}$ is reversible and its unique stationary distribution is $\pi$ according to the Metropolis-Hastings rule.
\end{proof}

\section{Proof of Theorem \ref{THM:MIXING}}

%The unicity of $K$-apps from a single record can easily be approximated with a standard Chernoff bound using uniform samples over all $K$-appss from the record. Notice that this sampling is trivial to implement (e.g., by choosing $K$ items from the record without replacement). 
In order to prove $\mathcal{M}$'s mixing time, we use a standard coupling argument which is described below.

\begin{definition}[Coupling]
\label{def:coupling}
A coupling of a Markov chain $\mathcal{M}$ on state space $\Omega$ is a Markov chain on $\Omega \times \Omega$ defining a stochastic process  $(X_t,Y_t)_{t=0}^{\infty}$ such that
\begin{itemize}
\item each of the processes $(X_t,\cdot)$ and $(\cdot,Y_t)$, viewed in isolation, is a faithful copy of the Markov chain $\mathcal{M}$ (given initial states $X_0 = x$ and $Y_0 = y$); that is, $Pr[X_{t+1} = b|X_t = a] = P_{\mathcal{M}}(a,b) = Pr[Y_{t+1} = b|Y_t = a]$;
and
\item if $X_t =Y_t$, then $X_{t+1} =Y_{t+1}$.
\end{itemize}
\end{definition}

Condition 1 ensures that each process, viewed in isolation, is just simulating the original chain $\mathcal{M}$, and the coupling is designed such that $X_t$ and $Y_t$ tend to coalesce (i.e., move closer to each other according to some notion of distance). Once they meet, Condition 2 guarantees that they will move together forward. The time of this coalescence can be used to upper bound the mixing time which is shown by the next lemma.

\begin{lemma}[Coupling lemma \cite{LevinPW09mixing}]
\label{lem:coupling}
Let $(X_t,Y_t)_{t=0}^{\infty}$ be a coupling of a Markov chain $\mathcal{M}$. 
For initial states $x,y$ let
$T^{x,y} = \min\{ t : X_t = Y_t| X_0 =x, Y_0 = y \}$
denote the random variable describing the time until $X_t$ and $Y_t$ coalesce. Then
$$
||P^t_{\mc{M}} - \pi||_{\mathit{tv}} \leq  \max_{x,y \in \Omega}Pr[T^{x,y} > t] 
$$
\end{lemma}

\begin{proof}[of Theorem \ref{THM:MIXING}]
Define a coupling $(X_t, Y_t)$ as follows. Let $X_t$ and $Y_t$ choose the same individual $u$ and subset $C$ in Line 6 and 7 of Algorithm \ref{alg:generic_chain}, respectively. This is a valid coupling according to Definition \ref{def:coupling}, since both $X_t$ and $Y_t$ are the exact copies of $\mathcal{M}$, and they move together after they coalesce.

Let $p(x) = \kappa(\cdot, x)$ denote the probability that $x=C$ is selected in Line 6 of Algorithm \ref{alg:generic_chain}.
Let $X_0=x$ and $Y_0 = y$, and, w.l.o.g., $p(x) \leq p(y)$. Due to the coupling rule, $X_t$ and $Y_t$ can coalesce at any time, since $P_\mathcal{M}(x,y) > 0$ for all $x,y \in \Omega$. This happens when both $X_t$ and $Y_t$ select a state $z$ such that $p(z) \leq p(x) \leq p(y)$, since $q(z) \leq q(x) \leq q(y)$ will also hold. Let $U_{\max} = \max_u U_u$. For any $x, z \in \Omega$, where $z$ occurs only in $U_{\max}$, $p(z) \leq p(x)$.
Indeed, 
\begin{align*}
\label{ineq:1}
p(x) &= \frac{1}{|U|}\sum_{\forall u : U_u \supseteq x}  \frac{1}{{|U_u| \choose K}} \\
%&\geq \frac{1}{|U|} \min_{u} \prod_{i=1}^K \frac{K!}{|U_u| - K +i}\\
%& \geq \frac{1}{|U|}\prod_{i=1}^K \frac{K!}{|U_{\max}| - K +i} \\
& \geq\frac{1}{|U|} \frac{1}{{|U_{\max}| \choose K}}\\
&\geq  p(z) 
\end{align*}

Hence, $X_t$ and $Y_t$ coalesce as soon as they select any  $z \in \Omega$ which occur only in the largest record in $D$. Therefore,
\begin{align}
||P^t_{\mc{M}} - \pi||_{\mathit{tv}} &\leq \max_{x,y \in \Omega^K}Pr[T^{x,y} > t]  \tag{by Lemma \ref{lem:coupling}} \\
&\leq \sum_{i=t}^{\infty} \left(1-H_1^*/ |U|\right)^i  H_1^*/|U| \notag \\
%&\leq 1 - \frac{H_1^*}{|U|} \sum_{i=0}^{t-1} \left(1-H_1^*/ |U|\right)^i   \notag \\
&\leq (1 - H_1^*/|U|)^t \notag \\
&\leq \exp\left(-tH_1^*/|U|\right) \notag
\end{align}
which proves the theorem.
\end{proof}

\balancecolumns
\end{document}